\documentclass[journal,lettersize]{IEEEtran}
\usepackage[utf8]{inputenc} 
\usepackage[T1]{fontenc}
\usepackage{url}
\usepackage{ifthen}
\usepackage{cite}
\usepackage[cmex10]{amsmath} % Use the [cmex10] option to ensure complicance
\usepackage{amsthm} 
\usepackage{algorithmic}
\usepackage{graphicx}
\usepackage{textcomp}
\usepackage{xcolor}
\usepackage{amssymb}
\usepackage{bbm}
\usepackage{bm}
\usepackage[linesnumbered,ruled]{algorithm2e}
\usepackage{subfigure}
\usepackage{marvosym}
\usepackage{booktabs}
\newtheorem{theorem}{Theorem}%[section]
\newtheorem{corollary}{Corollary}
\newtheorem{lemma}{Lemma}

\newtheorem{definition}{Definition}

\newtheorem{example}{Example}
\newtheorem{remark}{Remark}

\begin{document}
\IEEEoverridecommandlockouts
\title{Spatial Network Calculus: Toward Deterministic Wireless Networking} 
% Spatial Network Calculus: A Unified Framework for Deterministic Performance Guarantees
\author{Yi Zhong, \emph{Senior Member, IEEE}, Xiaohang Zhou, Ke Feng \\
\thanks{
Yi Zhong and Xiaohang Zhou are affiliated with the School of Electronic Information and Communications at Huazhong University of Science and Technology, Wuhan, China. Ke Feng started this work at INRIA-ENS Paris, France, and is now with CNRS, Laboratoire ETIS-UMR 8051, CY Cergy Paris Université, ENSEA, Cergy-Pontoise, France. 

This research was supported by the National Natural Science Foundation of China (NSFC) grant No. 62471193, and the Wuhan Science and Technology Program (No. 2024030803010178). The work of K. Feng was supported by the ERC NEMO grant, under the European Union's Horizon 2020 research
and innovation programme, grant agreement number 788851 to INRIA.
Preliminary findings of this work were presented in WiOpt 2024. 

The corresponding author is Yi Zhong (yzhong@hust.edu.cn).}
}

% \author{
% Xiaohang Zhou, Yi Zhong, \textit{Senior Member, IEEE}, Donglin Jia, Ke Feng

% \thanks{Xiaohang Zhou, Yi Zhong, Donglin Jia are with School of Electronic Information and Communications, Huazhong University of Science and Technology, Wuhan, P. R. China. Ke Feng is with INRIA-ENS, Paris, France. 

% This research was supported by the National Natural Science Foundation of China (NSFC) grant No. 62071190.

% The corresponding author is Yi Zhong (e-mail: yzhong@hust.edu.cn). 
% }
% }

\maketitle

\begin{abstract}
This paper extends the classical network calculus to spatial scenarios, focusing on wireless networks with {differentiated services} and varying transmit power levels. Building on a spatial network calculus, a prior extension of network calculus to spatial settings, we propose a generalized framework by introducing regulations for stationary marked point processes. The regulations correspond to two key constraints: the total transmit power {of all transmitters} within a spatial region and the cumulative received power at a receiver{, which we refer to as ball regulation and shot-noise regulation, respectively}. Then we prove the equivalence of ball regulation and shot-noise regulation for stationary marked point processes and establish a universal lower bound on the performance of all network links under these constraints. This framework is applicable to diverse network scenarios, as demonstrated by the analysis of performance guarantees for networks with multi-class users. In addition, we propose an SINR-based power control scheme adapted to user traffic, which ensures differentiated quality of service (QoS) for different user classes. 
% We derive deterministic performance guarantees for all links in complex and heterogeneous wireless networks.
\end{abstract}


\section{Introduction}
\subsection{Motivations}
Deterministic networking (DetNet), as a frontier in the field of data transmission, is transitioning from traditional \textit{best-effort} services to  \textit{guaranteed-performance} services paradigm \cite{8458130, 9454588}. This paradigm shift presents broad application prospects in critical areas such as smart manufacturing and vehicular networks. However, expanding the traditional (wired) deterministic networking framework to the wireless domain introduces a series of complex scientific challenges and technical bottlenecks due to the inherent characteristics of wireless communications. These challenges arise primarily from the electromagnetic properties of wireless links and the shared nature of the propagation medium. Unlike the controlled environment of wired networks, wireless propagation is vulnerable to multipath reflections and path obstructions, and the lack of physical isolation results in significant interference issues. Consequently, providing deterministic performance services in wireless network environments is particularly challenging.
In light of dynamic changes in service demands, randomness in channel conditions, unpredictability of resources, and potential network congestion, it is critical to develop new traffic management strategies, channel coding techniques, resource allocation schemes, and congestion control mechanisms. Such developments are crucial to ensure the reliability and controllable delay of data transmission.
%, especially in scenarios such as autonomous driving, industrial control, and telemedicine, where communication reliability is paramount. Therefore, research on deterministic wireless networks is not only urgent, but also essential for guaranteeing deterministic performance in dynamic and unpredictable wireless environments.

The theory of network calculus plays an indispensable role in the design of deterministic networks \cite{cruz1991calculus, 61110}. Initially developed for wired networks, 
%this theory employs precise mathematical modeling to 
network calculus introduces envelopes to describe network traffic and services, establishing deterministic upper bounds on network delay and data backlog. 
%Traditional network calculus focuses on assessing worst-case performance, providing deterministic bounds on delay, backlog, and other QoS parameters in wired networks. 
However, extending this theory to the wireless network domain presents several critical challenges. The intrinsic characteristics of wireless networks, such as signal attenuation, multipath propagation, and dynamic traffic variations, substantially increase the complexity of guaranteeing deterministic performance. Consequently, traditional network calculus methods are challenging to apply directly to wireless networks.

To address these challenges, 
%prior research introduced 
a \textit{spatial network calculus} is introduced in \cite{feng2023spatial}, which extends the classical network calculus by incorporating the spatial dimension along with the temporal dimension. This advancement offers a novel approach to providing deterministic performance guarantees on all links in wireless networks. The core concept involves using stochastic geometry \cite{haenggi2012stochastic, 6042301,BLASZCZYSZYN2021153} to analyze each wireless link in {spatially regulated} large-scale networks (spatial dimension) and integrating this with classical network calculus methods that focus on individual data packets (temporal dimension). 
{By applying spatial regulations, the number of simultaneously active transmitters is controlled, which is necessary to enable deterministic performance guarantees across all links in wireless networks.}
% This results in a new set of spatial regulations and methods to ensure deterministic performance in wireless networks. 
% This step shows significant potential for deterministic performance guarantees in wireless networks, particularly in scenarios requiring high reliability and real-time performance, such as autonomous driving, industrial control, and telemedicine.

{The spatial regulations introduced in \cite{feng2023spatial} focuses on constraining the number of transmitters for homogeneous networks. It does not capture the diversity of transmit power, traffic, and QoS demands in wireless networks with differentiated services \cite{10746553}.}
%, such as wireless power control, particularly in spatially correlated settings. 
%With the expansion of mobile applications, network nodes and traffic have become increasingly diversified , driving the demand for different quality of services (QoSs). 
For example, an approach is needed that tailors for both high-demand and low-demand users. Although next-generation mobile communication systems such as 6G have made strides in energy efficiency, ultra-reliability, and low latency, ensuring service guarantees for all users in large-scale heterogeneous wireless networks remains a critical and complex challenge \cite{hwang2022study, zhong2024toward}. 

% This paper focuses on addressing the diverse signal-to-interference-plus-noise ratio (SINR) requirements of different user classes in such heterogeneous environments.

This paper broadens the scope of the spatial network calculus in \cite{feng2023spatial} to provide performance guarantees in wireless networks with different power levels and types of traffic. Our approach provides deterministic performance guarantees while accommodating diverse QoS requirements. We offer both theoretical insights, enriching the spatial network calculus framework, and practical strategies for reliable network performance in diverse wireless environments. To meet the demands of complex traffic and limited resources, we propose and rigorously evaluate customized power control strategies, highlighting the importance of spatial regulation and strategic resource allocation to ensure QoS for all links.

{The proposed spatial network calculus applies to a wide range of real-world deployments with dense and heterogeneous wireless traffic. For instance, in factory automation or smart intersections, nodes are often clustered in confined areas, yet require reliable and low-latency communications. As long as active transmitters satisfy the proposed spatial regulations, our framework guarantees a lower bound on the link success probability for all links, ensuring predictable performance even in densely populated wireless environments.}

\subsection{Related Works}
% 现有工作总结
Using an appropriate point process is critical for accurately modeling wireless networks. {The poisson point process (PPP) model is commonly used to depict the spatial locations of network nodes \cite{6042301, zhong2020spatio}.
When the transmitters are modeled as a PPP and when they implement an access control protocol like Carrier Sense Multiple Access (CSMA), the effective transmitters are often modeled as a Matérn hardcore point process (MHCPP) \cite{nguyen2007stochastic, net:Haenggi11cl}.} In \cite{nguyen2013performance}, an extension of the MHCPP is used, where the competition is based on pair-wise SIR protection zones with the modification to the RTS/CTS mechanism which defines a receiver-centered protection zone to restrict other transmitters' activity within the area. Mishra et al. \cite{mishra2020stochastic} employ stochastic geometry to realistically model automotive radar interference, incorporating the Matérn point process to account for the spatial mutual exclusion of vehicles. {Di Renzo et al. \cite{di2018system} use the MHCPP for modeling base station locations within cellular networks. Chen {and Zhong} 
\cite{chen2024characterizing} analyze the stable packet arrival rate region in discrete-time slotted random access network by arranging the transmitters as an MHCPP.}

Most existing literature focuses on improving the performance of wireless networks including link success probability, rate, and latency. Miao {et al.} \cite{miao2012joint} propose a joint power control algorithm with QoS guarantee using the convex optimization method that satisfies the minimum rate constraint of the delay-constraint service and the proportional fairness of the best-effort service. Liu's research \cite{liu2011utility} proposes an algorithm that not only enhances the utility for best-effort  users but also diminishes the outage probability for rate-constrained users. 
In scenarios where resource availability is constrained, the allocation of priority becomes a critical consideration. Choi \cite{choi2022modeling} suggests a mechanism where vehicles engage in transmission based on the highest requirement among their neighboring vehicles. Kumar and colleagues \cite{kumar2018priority} introduce a method for assigning terminal requirements, factoring in both the ongoing service and the terminal location. {Abdel-Hadi and Clancy \cite{abdel2014utility} focus} on finding the optimal solution for the resource allocation problem that includes users with requirements. Lee's work \cite{lee2004opportunistic} delves into the realm of optimal power control, prioritizing services and thereby directing power control in alignment with these requirements. In \cite{kumar2018priority, abdel2014utility, lee2004opportunistic}, the authors focus on the multi-priority mechanisms and power resource allocation strategies in communication networks, which improve performance and fairness to terminal users.

However, the link performance of all may not be deterministically guaranteed previously due to uncontrolled transmitter locations. By contrast, we consider spatially regulated point processes, which is inspired by the deterministic network calculus framework \cite{chang2001performance,le2001network,cruz1991calculus}. The seminal work of Cruz in classical network calculus introduces the $(\sigma, \rho)$ traffic regulation, where $\sigma$ represents an allowable level of initial burstiness, {and $\rho$ sets the upper limit for the data arrival rate}. Thus, the traffic flow is bounded by a baseline constant $\sigma$ and an additional rate $\rho$ scales the duration of the time interval \cite{cruz1991calculus}. {Chang and Le Boudec have summarized advances in the study of queue lengths, packet loss rates, and traffic characteristics in the packet queuing process within deterministic network calculus \cite{chang2012performance, le2001network}.}
As a significant derivative branch of network calculus, stochastic network calculus provides precise quantitative methods for the statistical limits of network performance, marking a new milestone in the development of the theory \cite{6868978}. In their comprehensive review of stochastic network calculus, {Jiang et al. \cite{jiang2008stochastic}} introduce two key concepts: the stochastic arrival curve and the stochastic service curve. {These concepts have been fundamental in the study of service guarantee theories and have laid the theoretical groundwork for ensuring performance guarantees (stochastic or deterministic).} {Adamuz-Hinojosa et al. \cite{adamuz2022stochastic} build upon the stochastic network calculus framework to provide an upper bound on packet transmission delay for uRLLC RAN slices in a single cell.} In recent years, academic research on network calculus theory has also made significant advancements \cite{7033634, 10382447, 10412651, 10457574}. {Unlike stochastic network calculus\cite{10220199} and classical network calculus\cite{9796648}, which analyze network performance in the temporal domain, spatial network calculus provides deterministic performance guarantees by applying spatial regulations.}

Expanding these principles into spatial contexts, \cite{feng2023spatial} applies the parameters $\sigma$, $\rho$, and an additional parameter $\nu$ to govern the number of wireless links, the cumulative received power, and the interference within large-scale wireless networks. This methodology leads to the development of $(\sigma, \rho, \nu)$-ball regulation and $(\sigma, \rho, \nu)$-shot-noise regulation in spatial network calculus. {
%A preliminary version of this work appeared in \cite{wiopt24}, which focuses on these regulations for multi-priority traffic.
These regulations are applied to analyze wireless networks with multi-priority traffic in \cite{wiopt24}, which is a preliminary version of this work. In contrast, the present paper generalizes the earlier work through a novel framework based on stationary marked point processes and extends the regulations to ensure service guarantees for diverse network scenarios. Furthermore, we propose a new power control scheme tailored to manage power allocation in a bipolar network architecture \cite{baccelli2010stochastic}, where transmitters' spatial distribution follows a hardcore point process and user or traffic requirements differentiate across the network. % Finally, this paper presents substantial new theoretical results, detailed proofs, and refined analytical techniques, contributing to greater theoretical depth and academic rigor.
} 
% In this work, we extend these regulations to provide a methodology to ensure service guarantees for diverse and heterogeneous network scenarios. Furthermore, our research proposes a power control scheme. This scheme is specifically designed to manage the allocation of power in a bipolar network architecture \cite{baccelli2010stochastic}, where the transmitters follows a hardcore process, and takes into account the presence of users or traffic of varying requirements. 

% {A preliminary version of this work was presented in \cite{wiopt24}, where we investigated service guarantees for multi-priority traffic based on the spatial network calculus framework. In contrast to \cite{wiopt24}, this paper extends and generalizes the earlier work through a novel framework based on marked point processes. The regulation concepts have also been reformulated to broaden the applicability beyond the specific cases considered previously. Moreover, this paper introduces substantial new theoretical results, detailed proofs, and refined analytical techniques, offering greater theoretical depth and academic rigor.}

\subsection{Contributions} 
Our work extends the classical network calculus framework to spatial domains with generalized power and interference constraints, enhancing its applicability to practical and complex wireless network scenarios. The key contributions of this paper are as follows.

\begin{itemize} \item 
We extend two regulatory constraints {in the spatial network calculus}, i.e., ball regulation and shot-noise regulation, to stationary marked point processes and prove their equivalence in this context. These regulations provide a unified theoretical foundation for deriving performance lower bounds applicable to all network links in complex and heterogeneous wireless networks.
% \item We develop a generalized framework that provides deterministic guarantees on network performance, including lower bounds for link success probability under diverse user and traffic demands.
\item {We derive new lower bounds for link success probability under differentiated user and traffic demands.}
% \item To demonstrate the practical utility of the proposed framework, we analyze a traffic-adaptive power control strategy tailored to heterogeneous traffic conditions. Although this strategy serves as an illustrative example, it highlights how our results can guide power control to meet varying QoS requirements.
\item {To demonstrate the practical utility of the proposed framework, we analyze a traffic-adaptive power control strategy tailored to traffic demands. It highlights how our results can guide power control to meet varying QoS requirements.}
\end{itemize}

This paper is organized as follows. Section II introduces spatial regulations to govern wireless network configurations. Section III presents a heterogeneous traffic example and derives link success probability lower bounds for all links based on these regulations. Numerical results are discussed in Section IV, followed by conclusions in Section V.

% \section{System Model}
\section{Spatial Regulations}
For point processes without spatial regulations, such as the PPP, obtaining a lower bound for link success probability or guaranteeing deterministic performance for all links is not feasible due to unbounded interference from nearby transmitters.  
{In contrast, by applying spatial regulations, such as thinning a PPP into a hardcore point process, we can establish upper bounds for both the total transmitter power and the total received power within any area, enabling a link success probability lower bound across the network. In particular, introducing spatial exclusion between nodes serves as a special form of spatial regulation.} In this section, we will define ball regulation and shot-noise regulation for stationary marked point processes and explain their physical significance.

\subsection{Definitions}
To enhance a guaranteed lower bound of network performance, we introduce regulations on the transmit power. We first define the strong version of these regulations, which must be universally applicable across the entire space. {Then we introduce the weak version, which requires that the regulations are satisfied only from the observation of a jointly stationary point process.}

\begin{table}[ht]
\centering
\caption{{Summary of Notations}}
\label{tab:notations}
{
\begin{tabular}{c||l}
\hline
\textbf{Symbol} & \textbf{Definition} \\
\hline
$\widetilde{\Phi}$ & Stationary marked point process \\
$\Phi$ & Set of transmitter locations  \\
$\Psi$ & Jointly stationary point process of  $\widetilde{\Phi}$ \\
$\sigma,\,\rho,\,\nu$ & Non-negative spatial regulation parameter \\
$P_{x}$ & Transmit power of transmitter $x\in\Phi$ \\
$b(o, r)$ & The open ball centered at $o$ with radius $r$  \\
$P_{\rm total}(o, r)$ & Total transmit power of all transmitters within $b(o, r)$ \\
$N(r)$ & Total number of transmitters within $b(o, r)$ \\
$H$ & Hardcore distance \\
$\ell(\cdot)$ & Path loss function \\
$\alpha$ & Path loss exponent \\
$A_{\ell}$ & The upper bound for the total contribution of $\ell$ \\
$h_{x}$ & Small-scale fading coefficient from $x$ to the receiver \\
$\gamma_0$ & Link success probability lower bound \\
$M$ & Total number of user classes in the network \\
$P_\mathrm{s,min}$ & Minimum link success probability \\
${}_2F_1(\cdot)$ & Gaussian hypergeometric function \\
$\beta$ & Power allocation parameter \\
$r_{0}$ & Fixed distance from the transmitter to the receiver \\
\hline
\end{tabular}}
\end{table}

Consider a stationary marked point process \(\widetilde{\Phi} = \{(x, P_{x})\}\) on \(\mathbb{R}^2 \times \mathbb{R}^+\), defined on the probability space \((\Omega, \mathcal{A}, \mathbb{P})\). Here, \(\Phi  \triangleq  \{x : (x, P_{x}) \in \widetilde{\Phi}\}\) represents the set of transmitter locations and \(P_{x} > 0\) denotes the transmit power of the transmitter \(x\in\Phi\).\footnote{It is worth noting that marks can represent properties other than the transmit power, in which case the  results presented in this work remain valid.}
Let \( P_{\rm total}(y, r) \triangleq \sum_{x \in \Phi \cap b(y,r)} P_{x} \) denote the total power allocated to all transmitters within \(b(y,r)\), the open ball of radius \( r \) centered at \(y\in\mathbb{R}^2\). 
% If \(y\) is arbitrarily chosen in space, we obtain the strong regulations; if \(y\) is selected based on another jointly stationary point process, we obtain the weak regulations. 
In this work, we use the term ``a.s.'' (almost surely) to indicate that an event occurs with probability 1. {Table \ref{tab:notations}
summarizes the key notations defined in this paper.}

With these notations in place, we can now define the strong and weak $(\sigma, \rho, \nu)$-ball regulations for the stationary marked point process as follows.

\begin{definition}
\label{power-reg1}
(Strong $(\sigma, \rho, \nu)$-ball regulation). A stationary marked point process \(\widetilde{\Phi}\)
% $\Phi \triangleq \{x_i\}$ 
is said to exhibit strong $(\sigma, \rho, \nu)$-ball regulation if, for all \( r \geq 0 \),
\begin{equation}
    P_{\rm total}(o, r) \leq \sigma + \rho r + \nu r^2, \quad \mathbb{P}\text{-a.s.},
    \label{eq: def-1}
\end{equation}
where 
%\( P_{\rm total}(o, r) = \sum_{x \in \Phi \cap B(o,r)} P_{x} \) is the total power allocated to all transmitters within a circular area of radius \( r \)  centered at the origin denoted by \(B(o,r)\), and
\( \sigma, \rho, \nu \) are non-negative constants.
\end{definition}

Alternatively, one can write (\ref{eq: def-1}) as \(\mathbb{P}(P_{\rm total}(o, r) \leq \sigma + \rho r + \nu r^2) = 1\). Due to the stationarity of the marked point process, the probability distribution of \(P_{\rm total}(y, r)\) is identical to that of \(P_{\rm total}(o, r)\) for all \(y \in \mathbb{R}^2\) and \(r > 0\). {If \(y\) is arbitrarily chosen in space, we obtain the strong regulations; if \(y\) is selected based on another jointly stationary point process, we obtain the weak regulations.}

\begin{definition}\label{weak-power-reg1}
(Weak $(\sigma, \rho, \nu)$-ball regulation). A stationary marked point process \(\widetilde{\Phi}\) exhibits weak $(\sigma, \rho, \nu)$-ball regulation with respect to a jointly stationary point process \(\Psi\) if, for all \( r \geq 0 \),
\begin{equation}
    P_{\rm total}(o, r) \leq \sigma + \rho r + \nu r^2, \quad \mathbb{P}_{\Psi}^{o}\text{-a.s.},
\end{equation}
where 
% \( P_{\rm total}(o,r) = \sum_{x \in \Phi \cap B(o,r)} P_{x} \) represents the total power allocated to transmitters within a circular area of radius \( r \), centered a given point \(o\in\Psi\), and 
\( \sigma, \rho, \nu \) are non-negative constants, and $\mathbb{P}_{\Psi}^{o}$ denotes the Palm probability of $\Psi$.
\end{definition}

{Notably, if $\Phi$ is independent of $\Psi$, then we get strong regulations defined earlier.}

\begin{remark}
In Definitions \ref{power-reg1} and \ref{weak-power-reg1}, \(\sigma\) represents the maximum power level for all transmitters. \(\rho\) scales linearly with the radius \(r\) of the region, reflecting the contribution to power from the periphery. Meanwhile, \(\nu\) scales quadratically with \(r\), accounting for the increase in total power proportional to the size of the ball, which is also influenced by the repulsion between transmitters.
\end{remark}

\begin{remark}
When the transmit power of all transmitters in \(\Phi\) is equal and normalized to 1, i.e., \(P_{x} \equiv 1\) for all \(x\in\Phi\), we have \(P_{\rm total}(o,r) = \sum_{x\in \Phi \cap b(o,r)} 1 = \Phi(b(o,r))\), where $ \Phi(b(o,r))$ denotes the number of points in \(\Phi\) within \(b(o,r)\). {Eq (\ref{eq: def-1}) becomes}
% This leads to the inequality:
\begin{equation}
    \Phi(b(o,r)) \leq \sigma + \rho r + \nu r^2, \quad \mathbb{P}\text{-a.s.}
\end{equation}

In this case, the ball regulation for a stationary marked point process is consistent with the definition of ball regulation introduced in \cite{feng2023spatial}, which governs the total number of points in a point process within a specified circular region. {It is worth noting that this upper bound is not achievable in general.}
\end{remark}

We derive the superposition property of the proposed ball regulations.
\begin{lemma}
\label{lem:superposition}

Consider a stationary marked point process \(\widetilde{\Phi}\) constructed as the superposition of \(M\) jointly stationary marked subprocesses \(\widetilde{\Phi}_i\), such that \(\widetilde{\Phi} = \bigcup_{i=1}^M \widetilde{\Phi}_i\). If each subprocess \(\widetilde{\Phi}_i\) is strongly (or weakly) \((\sigma_i, \rho_i, \nu_i)\)-ball regulated, then the superposed marked point process \(\widetilde{\Phi}\) is also strongly (or weakly) \(\left({\sigma}, {\rho}, {\nu}\right)\)-ball regulated, where 
\begin{equation}
{\sigma} \triangleq \sum_{i=1}^{M} \sigma_i, \quad {\rho} \triangleq \sum_{i=1}^{M} \rho_i, \quad \text{and} \quad {\nu} \triangleq \sum_{i=1}^{M} \nu_i.
\end{equation}
\end{lemma}

\begin{proof}
% delete Let $\Phi_i(B(o,r))$ represent the number of transmitters from the \(i\)-th subprocess within \( B(o,r) \).

Denote by \( P_{{\rm total},i}(o,r) \) the total power of the transmitters of the \(i\)-th subprocess within $b(o,r)$. Since each subprocess \(\widetilde{\Phi}_i\) is strongly (or weakly) \((\sigma_i, \rho_i, \nu_i)\)-ball regulated, 
% {we obtain from Definition \ref{power-reg1} that}
there exists an upper bound on the total power for all transmitters in the \(i\)-th subprocess within \( b(o,r) \), i.e.,
\begin{equation}
P_{{\rm total},i}(o,r) \leq \sigma_i + \rho_i r + \nu_i r^2, \quad \mathbb{P}\text{-a.s.} \label{eqn:P_total_i1}
\end{equation}

{
% Considering the stationary marked point process \(\widetilde{\Phi}\), 
Let \( P_{\rm total}(o,r) \) denote the total power of transmitters across all subprocesses within $b(o,r)$, i.e.,}
\begin{equation}
    P_{\rm total}(o,r) \triangleq \sum_{i=1}^{M} P_{{\rm total},i}(o,r). \label{eqn:P_total2}
\end{equation}
% $P_{\rm total}(o,r) \triangleq \sum_{i=1}^{M} P_{{\rm total},i}(o,r). \label{eqn:P_total2}$

By substituting (\ref{eqn:P_total_i1}) into {(\ref{eqn:P_total2})}, we obtain
\begin{equation}
P_{\rm total}(o,r) \leq {\sigma} +{\rho} r + {\nu} r^2, \quad \mathbb{P}\text{-a.s.}
\end{equation}
%where \(\tilde{\sigma} \triangleq \sum_{i=1}^{M} \sigma_i\), \(\tilde{\rho} \triangleq \sum_{i=1}^{M} \rho_i\), and \(\tilde{\nu} \triangleq \sum_{i=1}^{M} \nu_i\).

%Thus, from the definition of strong (or weakly) ball regulation, we conclude that the marked point process \(\widetilde{\Phi}\) is strongly (or weakly) \(\left(\tilde{\sigma}, \tilde{\rho}, \tilde{\nu}\right)\)-ball regulated.
\end{proof}

{Using the definition of ball regulation, we derive the following lemma to characterize its properties in the context of hardcore point processes (HCPPs) marked with a constant transmit power. An HCPP is a spatial point process that enforces mutual exclusion among points by ensuring that no two transmitters are located within a minimum distance  \( H \) of each other, thus effectively modeling interference-avoiding deployments in wireless networks. An MHCPP is a widely used example of an HCPP and comes in two variants, type I and type II. In a type I MHCPP, all transmitters are initially distributed according to a PPP, and any point that falls within a hardcore distance \( H \) of another point is simply removed, regardless of any additional attributes. In contrast, a type II MHCPP introduces a random timestamp for each transmitter in a PPP. When two or more transmitters fall within the exclusion zone of one another, the one with the earliest timestamp is retained, while the others are discarded. This timestamp-based selection allows the type II MHCPPs to achieve a higher density of active transmitters compared to the type I, making them particularly suitable for modeling CSMA protocols where temporal coordination is used to avoid collisions and minimize interference.}

\begin{lemma}
\label{lem:hardcore}
An HCPP on \(\mathbb{R}^2\) with a hardcore distance \(H\), marked by a constant transmit power \(P\), is strongly \(\left(P, \frac{2\pi P}{\sqrt{12}H}, \frac{\pi P}{\sqrt{12}H^2}\right)\)-ball regulated.
\end{lemma}

\begin{proof}
Consider the densest packing scenario of an HCPP on \(\mathbb{R}^2\) (see Fig. \ref{fig-densestPacking}), where \(N(r)\triangleq \Phi(b(o,r))\) denotes the total number of transmitters within a radius region \(r\), and \(H\) represents the hardcore distance. The hardcore condition enforces a minimum distance \(H\) between any two transmitters, with the theoretical maximum density of a hardcore packing in \(\mathbb{R}^2\) being \(\frac{\pi}{\sqrt{12}}\).

The upper bound of \(N(r)\) for \(r \geq H\) is given by the area of a circle with radius \(r+H\), divided by the area required per transmitter:
\begin{equation}
\pi H^2 N(r) \leq \frac{\pi^2 (r + H)^2}{\sqrt{12}}.
\end{equation}

Expanding and rearranging this expression, and applying \(P_{\rm total}(o,r) = N(r)P\), we get:
\begin{equation}
    P_{\rm total}(o,r) \leq P + \frac{2\pi P}{\sqrt{12}H}r + \frac{\pi P}{\sqrt{12}H^2} r^2. \label{eqn:hardcorebound}
\end{equation}

Due to the stationarity, we obtain the result in the lemma.
\end{proof}

\begin{figure}
\centering
\includegraphics[width=0.4\textwidth]{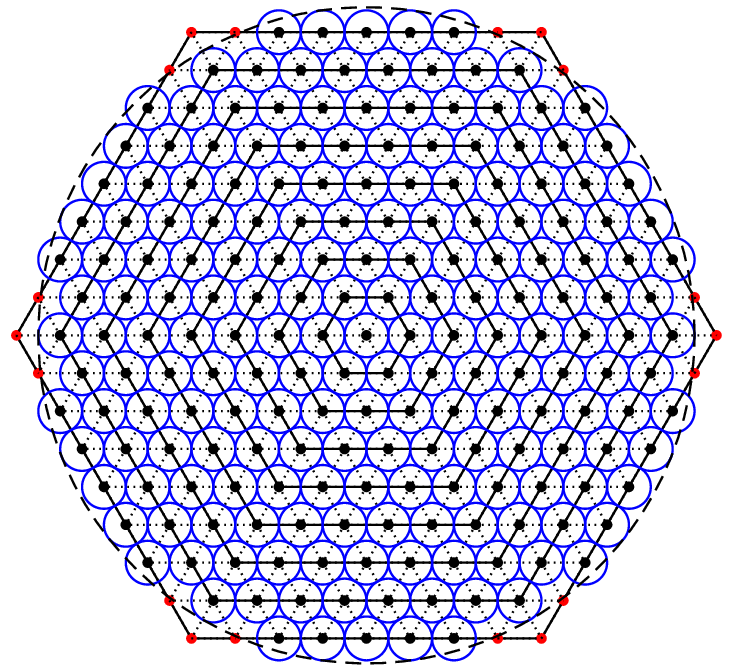}
\caption{Diagram of a triangular lattice with eight hexagonal rings of nodes in wireless network. Solid dots represent transmitters, blue circles represent guard zones with a radius of the hardcore distance $H$, and all retained transmitters are located within a dashed circle with radius $r$. Ring $k$ contains $6k$ nodes {for $k \geq 1$}\cite{haenggi2009interference}.}
\label{fig-densestPacking}
\end{figure}

% To prevent significant interference caused by assigning users in close proximity to the same $i$-th class, we could initially categorize users through independent and identically distributed sampling and divide areas to assess the proportion of users in the $i$-th class. If the proportion exceeds $p_i$, then the requirement of user is reduced. Conversely, if the proportion is below $p_i$, user previously of a lower requirement is elevated.

% \begin{lemma}
% Let $p_i$ be the proportion of the $i$-th class of users, any hardcore point process on $\mathbb{R}^2$ with hardcore distance $H$ is strongly $\left(p_i, 2\pi p_i/(\sqrt{12}H), \pi p_i/(\sqrt{12}H^2)\right)$-ball regulated.
% \end{lemma}
% \begin{proof}
% Considering the densest packing scenario, let $N(r)$ be the total number of transmitters for all user classes within a region of radius $r$, $H$ denote hardcore distance and $p_i$ indicates the proportion of the $i$-th class users. For cases where $r\geq H$ and the theoretical maximum density is $\pi/\sqrt{12}$, we have
% \begin{equation}\notag
%     \pi H^{2}N(r)\leq\pi^{2}(r+H)^{2}/\sqrt{12}
% \end{equation}
% where the upper bound of $N(r)$ can be simplified to $1+2\pi r/(\sqrt{12}H)+\pi r^2/(\sqrt{12}H^2)$, thus the hardcore point process is strongly $\left(p_i, 2\pi p_i/(\sqrt{12}H), \pi p_i/(\sqrt{12}H^2)\right)$-ball regulated.
% \end{proof}

\begin{remark}
Lemma \ref{lem:hardcore} states that an HCPP with a constant transmit power mark \(P\) is strongly ball regulated, with the upper bound specified in (\ref{eqn:hardcorebound}) holding for any circular area of radius \( r \) centered at an arbitrarily chosen point in space. If these centers are selected based on a jointly stationary point process \(\Psi\), the condition in (\ref{eqn:hardcorebound}) remains valid. Consequently, an HCPP with a hardcore distance \( H \) and constant transmit power mark \(P\) is also weakly $\left(P, \frac{2\pi P}{\sqrt{12}H}, \frac{\pi P}{\sqrt{12}H^2}\right)$-ball regulated with respect to any stationary point process.
\end{remark}

In addition to regulating total power within a circular region, we also impose regulations on the power of shot noise.

\begin{definition}\label{shot-noise-reg}
(Strong $(\sigma, \rho, \nu)$-shot-noise regulation). A stationary marked point process $\widetilde{\Phi}$ is strongly $(\sigma, \rho, \nu)$-shot-noise regulated if, for all non-negative, bounded, and non-increasing functions $\ell: \mathbb{R}^+\rightarrow \mathbb{R}^+$ and for all $R > 0$,
\begin{equation}\label{shot-noise-eq}
\begin{aligned}
    &\sum_{x\in\Phi\cap b(o,R)}P_{x}\ell(\|x\|)\leq \sigma\ell(0) \\
&\qquad+\rho\int_0^R\ell(r) \, {\rm d}r+2\nu\int_0^Rr\ell(r) \, {\rm d}r,\quad\mathbb{P}\text{-a.s.},
\end{aligned}
\end{equation}
where \(\|x\|\) denotes the Euclidean norm of the vector \(x\).
\end{definition}

By the stationarity of the marked point process \(\widetilde{\Phi}\), (\ref{shot-noise-eq}) holds for any \( y \in \mathbb{R}^2 \) as:
\begin{equation}
\begin{aligned}
    \sum_{x \in \Phi \cap b(y, R)} P_{x} \ell(\|x - y\|) \leq \sigma \ell(0) + \rho \int_0^R \ell(r) \, {\rm d}r \\
    \quad + 2\nu \int_0^R r \ell(r) \, {\rm d}r, \quad \mathbb{P}\text{-a.s.}
\end{aligned}
\end{equation}

If \( y \) is chosen from another point process \(\Psi\) rather than arbitrarily from the whole space \(\mathbb{R}^2\), we obtain the weak version of the shot-noise regulation, defined as follows.

\begin{definition}\label{weak-shot-noise-reg}
(Weak $(\sigma, \rho, \nu)$-shot-noise regulation). A stationary marked point process \(\widetilde{\Phi}\) is said to be weakly $(\sigma, \rho, \nu)$-shot-noise regulated with respect to \(\Psi\) if, for all non-negative, bounded, and non-increasing functions \(\ell: \mathbb{R}^+ \rightarrow \mathbb{R}^+\), for all \(R > 0\), and for all \( y \in \Psi \),
\begin{equation}\label{shot-noise-eq-weak}
\begin{aligned}
    \sum_{x \in \Phi \cap b(o, R)} P_{x} \ell(\|x\|) \leq \sigma \ell(0) + \rho \int_0^R \ell(r) \, {\rm d}r \\
    \quad + 2\nu \int_0^R r \ell(r) \, {\rm d}r,  \quad \mathbb{P}_{\Psi}^{o}\text{-a.s.}
\end{aligned}
\end{equation}
\end{definition}

When the transmit power is constant and normalized to 1, that is, \( P_{x} = 1 \) for all \( x \in \Phi \), the $(\sigma, \rho, \nu)$-shot-noise regulation reduces to the $(\sigma, \rho, \nu)$-shot-noise regulation presented in \cite{feng2023spatial}.

By considering the special case where the function {\(\ell(\cdot) \equiv 1\)}, we obtain a formulation analogous to the $(\sigma, \rho, \nu)$-ball regulation for all \(R > 0\), expressed as:
\begin{equation}
P_{\rm total}(o,R) \leq \sigma + \rho R + \nu R^2, \quad \mathbb{P}\text{-a.s.}
\end{equation}

\begin{remark}
As \( R \rightarrow +\infty \), (\ref{shot-noise-eq}) becomes:
\begin{equation}
    \sum_{x \in \Phi} P_{x} \ell(\|x\|) \leq A_{\ell}, \quad \mathbb{P}\text{-a.s.}, \label{eqn:R_inf_bound}
\end{equation}
where
\begin{equation}
    A_{\ell} \triangleq \sigma \ell(0) + \rho \int_0^\infty \ell(r) \, {\rm d}r + 2\nu \int_0^\infty r \ell(r) \, {\rm d}r \label{eqn:Ali}
\end{equation}
represents the upper bound on the total weighted sum of all points in \(\Phi\).
% , weighted by the function \(\ell\) and the transmit power \( P_{x} \). 
\end{remark}

{The physical meaning of \((\sigma, \rho, \nu)\)-shot-noise regulation is that it describes how power is regulated within a spatially distributed network.} Specifically, for a user at the origin, this regulation governs the cumulative impact of both useful signal power and interference from other transmitters. The total impact is bounded by the parameters \(\sigma\), \(\rho\), and \(\nu\). These constraints are functionally implemented through the path loss function \(\ell(\cdot)\), which models the decay of signal strength with distance in wireless networks. 

The relationship between ball regulation and shot-noise regulation is established in the following theorem, which states their equivalence.

\begin{theorem}
\label{thm:equ}
    A stationary marked point process \(\widetilde{\Phi}\) is strongly \((\sigma, \rho, \nu)\)-shot-noise regulated if and only if it is strongly \((\sigma, \rho, \nu)\)-ball regulated.
\end{theorem}

\begin{proof}
% Necessary Part:
To prove that strong \((\sigma, \rho, \nu)\)-shot-noise regulation implies strong \((\sigma, \rho, \nu)\)-ball regulation, we set the weight function \(\ell(\cdot) \equiv 1\) for all \(r \geq 0\) in (\ref{shot-noise-eq}). The shot-noise regulation simplifies to the total power allocated to all transmitters in \( b(o,r)\), bounded as:
\begin{equation}
    P_{\rm total}(o,r) = \sum_{x\in \Phi \cap b(o, r)} P_{x} \leq \sigma + \rho r + \nu r^2, \quad \mathbb{P}\text{-a.s.}
\end{equation}

Thus, strong \((\sigma, \rho, \nu)\)-ball regulation holds.

% Sufficient Part:
To prove that the strong \((\sigma, \rho, \nu)\)-ball regulation implies the strong \((\sigma, \rho, \nu)\)-shot-noise regulation, for $R>0$, we partition the open ball \(b(o, R)\) into \(n\) concentric annuli with radii \(r_k = k \Delta\), where \(\Delta \triangleq R / n\), \(n \in \mathbb{N}\), and \(k = 0, 1, \dots, n\) {(see Fig. \ref{fig-circleSegmantation})}. Let \(B_k \triangleq b(o, r_k)\), and denote the value of the weight function \(\ell(r)\) at radius \(r_k\) as \(l_k \triangleq \ell(r_k)\). The shot-noise generated by \(\ell\) within \(b(o, R)\) is
\begin{equation*}
\begin{aligned}
    &\sum_{x \in \Phi \cap b(o, R)} P_{x} \ell(\|x\|) \\
    &\overset{(\mathrm{a})}{=} \sum_{k=1}^{n} \sum_{x \in \Phi \cap (B_k \setminus B_{k-1})} P_{x} \ell(\|x\|) \\
    &\overset{(\mathrm{b})}{\leq} \sum_{k=1}^{n} \sum_{x \in \Phi \cap (B_k \setminus B_{k-1})} P_{x} l_{k-1} \\
    &\overset{(\mathrm{c})}{=} \sum_{k=1}^{n} l_{k-1} \left( P_{\rm total}(o,r_k) - P_{\rm total}(o,r_{k-1}) \right).
\end{aligned}
\end{equation*}

- {Step \(\mathrm{(a)}\): The sum over the circular region \(b(o, R)\) is partitioned into disjoint annular regions \(B_k \setminus B_{k-1}\).}

- Step \(\mathrm{(b)}\): The inequality follows from the monotonicity of \(\ell(r)\), which ensures that \(\ell(\|x\|) \leq l_{k-1}\) for \(x \in B_k \setminus B_{k-1}\).

- Step \(\mathrm{(c)}\): The summation over each annular region reduces to the difference in cumulative power between \(P_{\rm total}(o,r_k)\) and \(P_{\rm total}(o,r_{k-1})\).

As \(n \to \infty\), the summation converges to the Riemann-Stieltjes integral:
\begin{equation}
    \int_{0}^{R} \ell(r) \, {\rm d}P_{\rm total}(o,r),
\end{equation}
since \(P_{\rm total}(o,r)\) is bounded, non-decreasing, and integrable on \(\mathbb{R}^2\). This integral can be expanded using integration by parts:
\begin{equation*}
\begin{aligned}
    &\int_{0}^{R} \ell(r) \, {\rm d}P_{\rm total}(o,r) \\
    &= \ell(R) P_{\rm total}(o,R) - \ell(0) P_{\rm total}(o,0) -\int_{0}^{R} P_{\rm total}(o,r) \, {\rm d}\ell(r) \\   &\overset{(\mathrm{a})}{\leq} \ell(R) \left( \sigma + \rho R + \nu R^2 \right) - \int_{0}^{R} \left( \sigma + \rho r + \nu r^2 \right) \, {\rm d}\ell(r) \\
    &= \sigma \ell(0) + \rho \int_{0}^{R} \ell(r) \, {\rm d}r + 2 \nu \int_{0}^{R} r \ell(r) \, {\rm d}r, \quad \mathbb{P}\text{-a.s.}
\end{aligned}
\end{equation*}

In step \(\mathrm{(a)}\), the inequality uses the assumption that {$\ell(r)$ is non-negative and non-increasing , as well as that} \(\Phi\) is strongly \((\sigma, \rho, \nu)\)-ball regulated, bounding \(P_{\rm total}(o,R)\). The proof is complete.
\end{proof}

\begin{figure}
\centering
\includegraphics[width=0.4\textwidth]{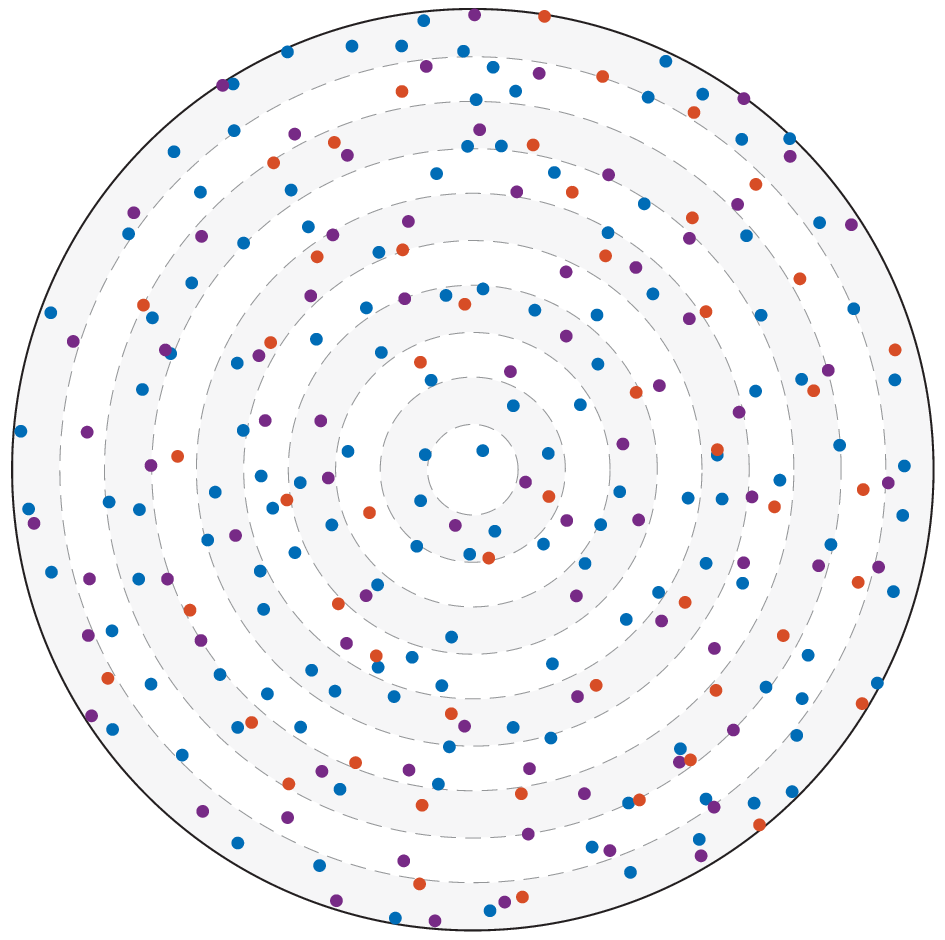} % 替换为实际文件名
\caption{
Illustration of a circular area with radius \(R = 15\), displaying the spatial distribution of transmitters operating at three distinct power levels. The circular region is divided into \(n = 10\) concentric annular regions, each with a spacing of \(\Delta = R/n = 1.5\). Transmitters are represented by solid dots, with orange, purple, and blue dots corresponding to high-power, medium-power, and low-power transmitters, respectively. The transmitter locations are initially generated using a PPP with an intensity \(\lambda = 0.3\), and then filtered to form a type II MHCPP. The resulting minimum separation distances (hardcore distances) are \(H_1 = 1\) for high-power transmitters, \(H_2 = 0.8\) for medium-power transmitters, and \(H_3 = 0.5\) for low-power transmitters. The dashed lines indicate the boundaries of the annular regions used for analytical segmentation in the proof process.
}
\label{fig-circleSegmantation}
\end{figure}

% \begin{figure}
% \centering
% \includegraphics[width=0.4\textwidth]{fig2.circleSegmantation.eps}
% \caption{Diagram of a circular area of radius \(R = 15\) showing transmitter locations of three different transmit power levels. The circular area is equally divided into \(n = 10\) annular regions with a spacing of \(\Delta = R/n = 1.5\). Orange, purple, and blue solid dots represent transmitters  corresponding to high-power, medium-power, and low-power transmitters, respectively. The transmitters are initially distributed according to a PPP with intensity \(\lambda = 0.3\) and then thinned to obtain a type II MHCPP. The corresponding hardcore distances are \(H_1 = 1\)  for high-power, \(H_2 = 0.8\) for medium-power, and \(H_3 = 0.5\) for low-power transmitters.}
% \label{fig-circleSegmantation}
% \end{figure}

Similarly, we get the following corollary, which states the equivalence for weak regulations. 
\begin{corollary}
\label{cor:equ}
    A stationary marked point process \(\widetilde{\Phi}\) is weakly \((\sigma, \rho, \nu)\)-shot-noise regulated with respect to a jointly stationary point process \(\Psi\) if and only if it is weakly \((\sigma, \rho, \nu)\)-ball regulated with respect to \(\Psi\).
\end{corollary}
\begin{proof}
    Apply the techniques in the proof of Theorem \ref{thm:equ}.
\end{proof}

\subsection{Performance Bounds}
In our study, we focus on two channel models: Rayleigh fading and the absence of fading. The analysis can be easily generalized to general fading \cite{feng2023spatial}. Without loss of generality, we consider a receiver located at the origin, with its dedicated transmitter positioned at \(x_0\). Let \(r_0\) denote the distance from the receiver to its dedicated transmitter \(x_0\), i.e., \(r_0=\|x_0\|\).

Using a bounded, non-increasing, continuous, and integrable function \(\ell(r)\) as the path loss function, we can express the SINR for the receiver as
\begin{equation}
    \mathrm{SINR} = \frac{P_{x_0} h_{x_0} \ell(r_0)}{I + W},
\end{equation}
where \(P_{x_0}\) is the transmit power of the transmitter at \(x_0\), \(W\) is the variance of the additive white Gaussian noise, and \(I\) is the interference power from other transmitters in the network expressed as
\begin{equation}
\begin{aligned}
I &= \sum_{x \in \Phi \setminus \{x_0\}}  P_{x} h_{x} \ell(\|x\|).
\end{aligned}\label{eqn:interwithfading_1tier}
\end{equation}

We denote \(h_x\) as the small-scale fading coefficient from the transmitter located at position \(x\) to {the user at the origin.}
% its corresponding receiver. 
In the absence of fading, we set \(h_x = 1\) for all transmitters, which simplifies the SINR expression. When considering channel fading, we adopt a Rayleigh fading model, where the fading coefficients are independent and identically distributed (i.i.d.). In this case, the power attenuation \(h_x\) follows an exponential distribution with unit mean, i.e., \(h_x \sim \text{Exp}(1)\).

\begin{theorem}
\label{thm:interwithoutfading}
When a stationary marked point process \(\widetilde{\Phi}\) is either strongly \((\sigma, \rho, \nu)\)-ball regulated or weakly \((\sigma, \rho, \nu)\)-ball regulated with respect to the receiver point process, and in the absence of channel fading (i.e., \(h_{x} \equiv 1\)), the interference conditioned on the transmit power \(P_{x_0}\) at \(x_0\) can be almost surely bounded as
\begin{equation}
    I \leq A_{\ell} - P_{x_0} \ell(r_0), \quad \mathbb{P}\text{-a.s.} \label{eqn:interwithoutfading_1tier}
\end{equation}

Consequently, an almost sure lower bound for the SINR, conditioned on the transmit power \(P_{x_0}\), is given by
\begin{equation}
    \mathrm{SINR} \geq \frac{P_{x_0} \ell(r_0)}{A_{\ell} - P_{x_0} \ell(r_0)}, \quad \mathbb{P}\text{-a.s.} \label{eqn:withoutfading_1tier}
\end{equation}
% \begin{equation}
%     \mathrm{SINR} \geq \frac{\theta P_{x_0} \ell(r_0)}{A_{\ell} - P_{x_0} \ell(r_0)}, \quad \mathbb{P}\text{-a.s.} \label{eqn:withoutfading_1tier}
% \end{equation}
\end{theorem}
\begin{proof}
Since the transmitter marked point process \(\widetilde{\Phi}\) is \((\sigma, \rho, \nu)\)-ball regulated, and due to the equivalence between ball regulation and shot-noise regulation established in Theorem \ref{thm:equ} and Corollary \ref{cor:equ}, we conclude that the marked point process \(\widetilde{\Phi}\) is also \((\sigma, \rho, \nu)\)-shot-noise regulated.

In a scenario without channel fading (that is, \(h_{x} = 1\)), {the interference experienced by a user can be bounded through (\ref{eqn:R_inf_bound})}.
Given this upper bound on interference, we can further derive almost sure lower bounds for the SINR in (\ref{eqn:withoutfading_1tier}). 
\end{proof}

{Theorem \ref{thm:interwithoutfading} provides an upper bound on interference in an integral form under the assumption of no fading. By assuming the path loss function \(\ell(r)=\min\{1, r^{-\alpha}\}\), where \(\alpha\) is the path loss exponent, we derive an explicit mathematical expression of the upper bound of interference in this special case.
\begin{example}
Considering the special case where \(\ell(r)=\min\{1, r^{-\alpha}\}\) with \(\alpha > 2\), the interference is upper bounded by
\begin{equation}
	\label{I-upper-bound}
	I \leq \sigma + \rho\frac{\alpha}{\alpha - 1} + \nu\frac{\alpha}{\alpha - 2} - P_{x_0}\min\{1, r_0^{-\alpha}\}, \quad \mathbb{P}\text{-a.s.}
\end{equation}
\end{example}}
% The proof is omitted here. }

Considering the case with channel fading, for a target SINR thresholds $\theta> 0$, the link success probability is
\begin{equation}
    P_{\mathrm{s}}(\theta \mid \widetilde{\Phi}) \triangleq \mathbb{P}(\mathrm{SINR} > \theta \mid \widetilde{\Phi}), \label{eqn:linkreliability_1tier}
\end{equation}
where $P_{\mathrm{s}}(\theta \mid \widetilde{\Phi}) \in [0, 1]$ is a random variable conditioned on the locations and the transmit powers of all transmitters. Link success probability can be interpreted as the probability that the SINR of a link with its user located at the origin exceeds the threshold $\theta$, given a network and power control realization $\widetilde{\Phi}$, and accounting for the effects of fading. The distribution associated with $P_{\mathrm{s}}(\theta \mid \widetilde{\Phi})$ is known as the \textit{meta distribution} \cite{haenggi2015meta, feng2020separability}.

Although the link success probability \(P_{\mathrm{s}}(\theta \mid \widetilde{\Phi})\) is subject to randomness due to fading, a definite lower bound \(\gamma_0\) can be established such that
\begin{equation}
    \mathbb{P}\left(P_{\mathrm{s}}(\theta \mid \widetilde{\Phi}) > \gamma_0\right) = 1.
\end{equation}

This value, \(\gamma_0\), is termed the success probability lower bound for all links. Similarly, if we require that \(80\%\) of the links meet the SINR threshold requirement, it can be expressed as
\(\mathbb{P}\left(P_{\mathrm{s}}(\theta \mid \widetilde{\Phi}) > \gamma_0\right) = 0.8\).  {Notably, in the absence of fading, the link success probability \(P_{\mathrm{s}}(\theta \mid \widetilde{\Phi}) \in \{0,1\}\).}  

The challenge then is to determine how to achieve a deterministic lower bound \(\gamma_0\) for all links within a large-scale wireless network through our proposed regulatory measures.

% \subsection{Properties}

% 1. (Superposition) Each class of users is subject to the $\beta$-threshold regulation, the power of all classes summed together are subject to a total threshold regulation, which is $\sum\beta$-ball regulated.

% 2. (Displacement) If the network is ball regulated and the amount of power change is bounded and stationarity-preserving, then the network after the power change is still ball regulated.

\begin{theorem}
\label{thm:boundnofading_1tier}
Considering Rayleigh fading, if a stationary marked point process \(\widetilde{\Phi}\) is {strongly} \((\sigma, \rho, \nu)\)-shot-noise regulated, then the link success probability given the transmit power \(P_{x_0}\) at \(x_0\) satisfies
\begin{equation}
    P_{\mathrm{s}}(\theta \mid \widetilde{\Phi}) \geq \exp\left(-\frac{\theta (W + A_{\ell})}{P_{x_0} \ell(r_0)} + \theta\right), \quad \mathbb{P}\text{-a.s.}, \label{eqn:successprob_withfading}
\end{equation}
where \(A_{\ell}\) is defined in (\ref{eqn:Ali}).
\end{theorem}

\begin{proof}
For users with a target SINR threshold \(\theta > 0\) in Rayleigh fading channels, the link success probability as given by (\ref{eqn:linkreliability_1tier}) is expressed as
\begin{equation}
\begin{aligned}
P_\mathrm{s}(\theta \mid \widetilde{\Phi}) 
&\overset{\mathrm{(a)}}{=} \mathbb{P}\left[\frac{P_{x_0} h_{x_0} \ell(r_0)}{I + W} > \theta \Bigg| \widetilde{\Phi}\right] \\
&\overset{\mathrm{(b)}}{=} \mathbb{E}\left[\exp\left(-\frac{\theta (I + W)}{P_{x_0} \ell(r_0)}\right) \Bigg| \widetilde{\Phi}\right],
\end{aligned}
\end{equation}
where \(\mathrm{(a)}\) follows from the definition of SINR, and \(\mathrm{(b)}\) utilizes the property that \(h_x\), the small-scale fading factor, is exponentially distributed. 
Here, \((x_0, P_{x_0}) \in \widetilde{\Phi}\), where \(x_0\) represents the designated transmitter for the receiver located at the origin, and \(P_{x_0}\) denotes its transmit power.

% The link success probability can be derived as
{By substituting the interference under fading given in (\ref{eqn:interwithfading_1tier}), the link success probability can be derived as}
\begin{equation*}
\begin{aligned}
&P_\mathrm{s}(\theta \mid \widetilde{\Phi}) = \exp\left(-\frac{\theta W}{P_{x_0} \ell(r_0)}\right) \\
&\times \mathbb{E}\left[\exp\left(-\frac{\theta \sum_{(x,P_{x}) \in \widetilde{\Phi} \setminus (x_0,P_{x_0})} P_{x} h_{x} \ell(\|x\|)}{P_{x_0} \ell(r_0)}\right)\right] \\
&\overset{{\mathrm{(a)}}} {=} \exp\left(-\frac{\theta W}{P_{x_0} \ell(r_0)}\right) \prod_{(x,P_{x}) \in \widetilde{\Phi}  \setminus (x_0,P_{x_0})} \frac{1}{1 + \frac{\theta P_{x} \ell(\|x\|)}{P_{x_0} \ell(r_0)}} \\
&= \exp\bigg(-\frac{\theta W}{P_{x_0} \ell(r_0)} -\!\!\!\!\!\!\!\!\!\!\!\!\!\! \sum_{(x,P_{x}) \in \widetilde{\Phi}  \setminus (x_0,P_{x_0})}\!\!\!\!\!\!\!\!\!\!\!\!\!\! \ln\left(1 + \frac{\theta P_{x} \ell(\|x\|)}{P_{x_0} \ell(r_0)}\right)\bigg) \\
&\overset{{\mathrm{(b)}}} {\geq} \exp\bigg(-\frac{\theta W}{P_{x_0} \ell(r_0)} - \frac{\theta}{P_{x_0} \ell(r_0)} \!\!\!\!\!\!\!\sum_{(x,P_{x}) \in \widetilde{\Phi}  \setminus (x_0,P_{x_0})} \!\!\!\!\!\!\!\!\!\!\!\!\!\!\!\!\! P_{x} \ell(\|x\|)\bigg) \\
&\overset{{\mathrm{(c)}}} {\geq} \exp\left(-\frac{\theta W}{P_{x_0} \ell(r_0)} - \frac{\theta (A_{\ell} - P_{x_0} \ell(r_0))}{P_{x_0} \ell(r_0)}\right) \\
&= \exp\left(-\frac{\theta (W + A_{\ell})}{P_{x_0} \ell(r_0)} + \theta\right), \quad \mathbb{P}\text{-a.s.},
\end{aligned}
\end{equation*}
{where \(\mathrm{(a)}\) follows from the independence assumption of fading and the assumption that \(\mathbb{E}[h_x]=1\), \(\mathrm{(b)}\) follows from the inequality \(\ln(1 + x) < x\) for \(x > 0\), and \(\mathrm{(c)}\) follows from the asymptotic upper bound derived in (\ref{eqn:R_inf_bound}). Expanding and rearranging, we get the link success probability lower bound in (\ref{eqn:successprob_withfading}).}
\end{proof}

The result established in the theorem highlights the impact of shot-noise regulation on link success probability in Rayleigh fading channels. The lower bound derived for the SINR demonstrates that effective ball regulation and interference management can significantly enhance network performance, even in the presence of fading. 

\section{Deterministic Performance Guarantees with Heterogeneous Traffic}
In an irregularly deployed wireless network, links may be in close proximity, leading to significant interference and degraded performance. To ensure QoS for diverse users, we can impose a minimum distance between nodes during deployment or implement medium access protocols such as CSMA to isolate links spatially. 
The key question is how to establish such isolation using the regulations proposed in the previous section to ensure the required performance.

Assume that there are \(M\) classes of links in a large wireless network, each with different target SINR thresholds \(\theta_1, \theta_2, \ldots, \theta_M\), where \(\theta_1 > \theta_2 > \ldots > \theta_M > 0\). We denote the type II MHCPP corresponding to the \(i\)-th link class by a stationary and ergodic point process \(\Phi_i\) on \(\mathbb{R}^2\), with a hardcore distance of \(H_i\). The superposition of point processes for different link classes is defined as \(\Phi \triangleq \bigcup_{i=1}^M \Phi_i\). It is important to note that the point processes \(\Phi_i\) are assumed to be independent of each other.

Let \(\Psi_i\) denote the point process corresponding to the receivers, which is jointly stationary and ergodic with \(\Phi_i\). The link distance \(r_0\) between each receiver and its dedicated transmitter is fixed. For simplicity, we assume that the bandwidth is normalized to 1. The transmit power of all transmitters within the \(i\)-th link class is constant and denoted by \(P_i\).

In the following discussion, we consider mutual exclusion among transmitters and employ the MHCPP model (see Fig. \ref{fig-bipolar}). {This model enforces mutually exclusive regulations for transmitter deployment, ensuring that the distance between any two transmitters belonging to the \(i\)-th link class is at least the hardcore distance \(H_i\).}

We utilize a type II MHCPP, which is more practical for modeling CSMA protocols compared to type I. 
% In a type I MHCPP, conflicting nodes are removed simultaneously, whereas in a type II MHCPP, each link is assigned a random timestamp, and the link with the earlier timestamp is retained. This mechanism allows the type II MHCPP to achieve a higher average density, thereby accommodating more users compared with the type I MHCPP. 
We will further explore the relationship between the mutually exclusive of the MHCPP and the regulations introduced in the previous section.

In our analysis, we examine all links within large-scale wireless networks characterized by multi-class traffic. The locations of transmitters are modeled as independent MHCPPs, thinned from PPPs with different hardcore distances. We assume that the retained links in the network are always active, enabling us to derive worst-case bounds without considering the impact of queueing. Our primary interest is whether we can guarantee the requirements of different classes of links through the introduction of regulations.

\begin{figure*}
\centerline{\includegraphics[width=1\textwidth]{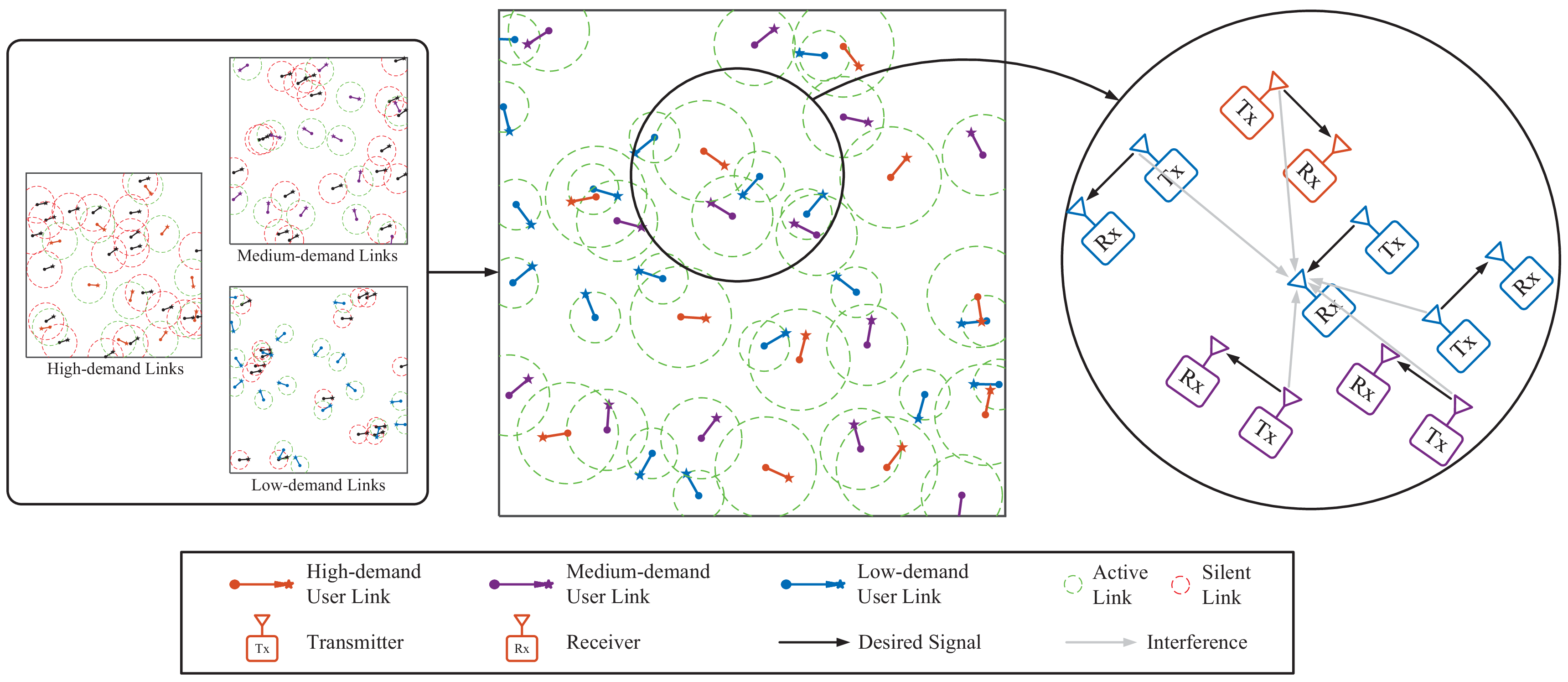}}
% 修改为多个不同硬核距离的硬核点过程叠加
\caption{Diagram of a bipolar network. Solid dots within dashed circles represent transmitters, distributed as a PPP with intensity $\lambda = 0.3$. Stars mark the uniformly distributed receiver locations {with three user classes}. Transmitters and receivers are separated by a fixed distance $r_0 = 0.5$. The network employs an MHCPP of type II for the different demand links respectively, thinning the PPP where retained transmitters are shown in green dashed circles. Active links are color-coded to indicate user SINR thresholds: orange links represent high threshold users with hardcore distance $H_1=1$, purple links represent medium threshold users with $H_2=0.8$ and blue for low threshold users with $H_3=0.5$.}
\label{fig-bipolar}
\end{figure*}

In this section, we hypothesize that the channels within the small-scale wireless network models are characterized by two cases: Rayleigh fading and the absence of fading. Using a bounded, non-increasing, continuous, and integrable function \(\ell(r)\) as the path loss function, we can express the SINR for the \(i\)-th link class as:
\begin{equation}
    \mathrm{SINR}_i = \frac{P_i h_{x_0} \ell(r_0)}{I + W},
\end{equation}
where \(I\) represents the interference power from other signals in the network, including contributions from transmitters in the \(i\)-th class as well as those in other classes.
\( h_{x_0} \) denotes the small-scale fading coefficient from a transmitter located at position \( x_0 \) to its corresponding receiver. 

\subsection{Performance Guarantee}
Since the MHCPP is a stationary point process, we can place the user of the \(i\)-th class at the origin without loss of generality and denote its dedicated transmitter by \(x_0 \in \Phi_i\). In a network with \(M\) classes of users, the total interference received by the user at the origin can be expressed as:
\begin{equation}
\begin{aligned}
% I &= \sum_{i=1}^M \sum_{x_i \in \Phi_{i} \setminus \{x_0\}} P_{i} h_{x_i} \ell(\|x_i\|) \\
I_i  &= \sum_{j=1}^M \sum_{x \in \Phi_{j}}  P_{j} h_{x} \ell(\|x\|) -P_{x_0}h_{x_0}\ell(r_0).
\end{aligned}
\end{equation}
% where \(h_{x_i}\) denotes the power of the small-scale fading from \(x_i \in \Phi_i\) to the origin.

\begin{theorem}
\label{thm:boundnofading}
In a heterogeneous scenario modeled by \(M\) MHCPPs and in the absence of channel fading (i.e., \(h_{x} \equiv 1\)), the interference experienced by a user in the \(i\)-th class can be bounded as
\begin{equation}
    I_i \leq \sum_{j=1}^M A_{\ell,j} - P_i \ell(r_0), \label{eqn:interwithoutfading}
\end{equation}
where the lower bound for the SINR experienced by a user in the \(i\)-th class is given by
\begin{equation}
    \mathrm{SINR}_i \geq \frac{P_i \ell(r_0)}{\sum_{j=1}^M A_{\ell,j} - P_i \ell(r_0)}, \label{eqn:withoutfading}
\end{equation}
and \(A_{\ell,j}\) represents the contribution from the \(j\)-th class to the overall interference, defined as
\begin{equation}
    A_{\ell,j} \triangleq \sigma_j \ell(0) + \rho_j \int_0^\infty \ell(r) \, {\rm d}r + 2\nu_j \int_0^\infty r \ell(r) \, {\rm d}r,
\end{equation}
where \(A_{\ell,j}\) provides an upper bound on the total weighted sum of all points in \(\Phi_j\), weighted by the path loss function \(\ell\) and the transmit power \(P_j\). Here, \(\sigma_j = P_j\), \(\rho_j = \frac{2\pi P_j}{\sqrt{12}H_j}\), and \(\nu_j = \frac{\pi P_j}{\sqrt{12}H_j^2}\). 
\end{theorem}

\begin{proof}
Since the transmitters of the \(j\)-th class of links form a type II MHCPP with a hardcore distance \(H_j\) and constant transmit power \(P_j\), applying Lemma \ref{lem:hardcore} shows that the stationary marked point process \(\widetilde{\Phi}_j\triangleq\{(x,P_{x}):x\in\Phi_j, P_{x}=P_j\}\) is strongly \(\left(P_j, \frac{2\pi P_j}{\sqrt{12}H_j}, \frac{\pi P_j}{\sqrt{12}H_j^2}\right)\)-ball regulated. Due to the equivalence between strong ball regulation and strong shot-noise regulation established in Theorem \ref{thm:equ}, we conclude that the stationary marked point process \(\widetilde{\Phi}_j\) is also strongly \(\left(P_j, \frac{2\pi P_j}{\sqrt{12}H_j}, \frac{\pi P_j}{\sqrt{12}H_j^2}\right)\)-shot-noise regulated.

In a scenario without channel fading (i.e., \(h_{x} \equiv 1\)), the interference experienced by a user in the \(i\)-th class can be bounded using (\ref{eqn:R_inf_bound}) as (\ref{eqn:interwithoutfading}).
Given this upper bound on interference, we can derive further lower bounds for link SINR in (\ref{eqn:withoutfading}). 
\end{proof}

The theorem emphasizes the critical relationship between interference and SINR in heterogeneous wireless networks modeled by multiple MHCPPs. In particular, the proof can also be approached from a different perspective. 

Consider that the transmitters associated with the \(i\)-th link class are distributed as a type II MHCPP \(\Phi_i\) on \(\mathbb{R}^2\) with a minimum separation distance (hardcore distance) \(H_i\). Since we have established that each marked MHCPP is \(\left(P_j, \frac{2\pi P_j}{\sqrt{12}H_j}, \frac{\pi P_j}{\sqrt{12}H_j^2}\right)\)-ball regulated in Lemma \ref{lem:hardcore}, we can invoke the lemma below to demonstrate that the superposition of the \(M\) marked MHCPPs, defined as \(\widetilde{\Phi} = \bigcup_{i=1}^M \widetilde{\Phi}_i\), is also a ball regulated point process.

\begin{lemma}
\label{lem:superreg}
The superposition point process \(\widetilde{\Phi}\) of \(M\) marked MHCPPs is strongly \(\left({\sigma}, {\rho}, {\nu}\right)\)-ball regulated, 
where \({\sigma} = \sum_{i=1}^{M} P_i\), \({\rho} = \sum_{i=1}^{M} \frac{2\pi P_i}{\sqrt{12}H_i}\), and \({\nu} = \sum_{i=1}^{M} \frac{\pi P_i}{\sqrt{12}H_i^2}\).
\end{lemma}

\begin{proof}
From Lemma \ref{lem:hardcore}, we know that the stationary marked point process with all points construing a type II MHCPP \(\Phi_i\) with hardcore distance \(H_i\) and with transmit power mark \(P_i\) is strongly \(\left(P_i, \frac{2\pi P_i}{\sqrt{12}H_i}, \frac{\pi P_i}{\sqrt{12}H_i^2}\right)\)-ball regulated.

Applying Lemma \ref{lem:superposition}, we conclude that the superposed marked point process \(\widetilde{\Phi}\) is strongly \(\left({\sigma}, {\rho}, {\nu}\right)\)-ball regulated.
\end{proof}

Thus, utilizing Theorem \ref{thm:boundnofading_1tier}, this alternative approach also confirms the results presented in Theorem \ref{thm:boundnofading}.

From Lemma \ref{lem:superreg}, we derive that the total power allocated to all transmitters within \(b(o,r)\) can be expressed as
\begin{equation}\label{totalPower}
P_\mathrm{total}(o,r) \leq \sum_{i=1}^M P_{i} + \sum_{i=1}^M \frac{2\pi}{\sqrt{12}} \left(\frac{P_{i}}{H_i}\right) r + \sum_{i=1}^M \frac{\pi}{\sqrt{12}} \left(\frac{P_{i}}{H_i^2}\right) r^2.
\end{equation}

This inequality provides an upper limit on the total transmission power of the network within \(b(o,r)\).

\begin{theorem}
\label{thm:linksuccbound}
In a heterogeneous scenario modeled by \(M\) MHCPPs and considering Rayleigh fading, the link success probability \(P_{\mathrm{s}}(\theta_i \mid \widetilde{\Phi})\) for the \(i\)-th class satisfies
\begin{equation}
    P_\mathrm{s}(\theta_i \mid \widetilde{\Phi}) \geq \exp\Bigg(-\frac{\theta_i W}{P_i \ell(r_0)} - \sum_{j=1}^M A_{\tilde{\ell}_{ij}}\Bigg), \quad \mathbb{P}\text{-a.s.}, \label{eqn:Ps_bound1}
\end{equation}
where \(A_{\tilde{\ell}_{ij}}\) is defined as
\begin{equation}
    A_{\tilde{\ell}_{ij}} \triangleq \sigma_j \tilde{\ell}_{ij}(0) + \rho_j \int_0^\infty \tilde{\ell}_{ij}(r) \, {\rm d}r + 2\nu_j \int_0^\infty r \tilde{\ell}_{ij}(r) \, {\rm d}r, \label{eqn:Ali_ij}
\end{equation}
and \(\tilde{\ell}_{ij}(r)\) is given by
\begin{equation}
    \tilde{\ell}_{ij}(r) \triangleq \ln\left(1 + \frac{\theta_i P_j \ell(r)}{P_i \ell(r_0)}\right). \label{eqn:tilde_ell_ij}
\end{equation}
\end{theorem}

\begin{proof}
Considering the case with Rayleigh fading channels and different target SINR thresholds \(\theta_i > 0\) for each class of links, the link success probability for the \(i\)-th class can be expressed as
\begin{equation}
\begin{aligned}
P_{\mathrm{s}}(\theta_i \mid \widetilde{\Phi}) &\triangleq \mathbb{P}(\mathrm{SINR}_i > \theta_i \mid \widetilde{\Phi}) \\
&= \mathbb{P}\left[\frac{P_i h_{x_0} \ell(r_0)}{I + W} > \theta_i \Bigg| \widetilde{\Phi}\right],
\label{eqn:linkreliability}
\end{aligned}
\end{equation}
where \(\widetilde{\Phi} = \bigcup_{i=1}^M \widetilde{\Phi}_i\) represents the superposition of stationary marked point processes for all \(M\) classes of links.

The link success probability for the \(i\)-th class user in a large-scale wireless network can be calculated as follows.
\begin{equation}
\begin{aligned}
&P_\mathrm{s}(\theta_i \mid \widetilde{\Phi}) = \exp\left(-\frac{\theta_i W}{P_{i} \ell(r_0)}\right) \\
& \times \mathbb{E}\left[\exp\left(-\frac{\theta_i \sum_{j=1}^M \sum_{x \in \widetilde{\Phi}_{j} \setminus \{x_0\}} P_{j} h_{x} \ell(\|x\|)}{P_{i} \ell(r_0)}\right) \Bigg| \widetilde{\Phi}\right]. \\
\end{aligned}
\end{equation}

This can be further simplified to
\begin{equation*}
\begin{aligned}
&\quad P_\mathrm{s}(\theta_i \mid \widetilde{\Phi}) = \exp\left(-\frac{\theta_i W}{P_{i} \ell(r_0)}\right) \\
&\qquad\qquad\times\prod_{j=1}^M \prod_{(x,P_{x}) \in \widetilde{\Phi}_j  \setminus (x_0,P_{x_0})} \frac{1}{1 + \frac{\theta_i P_{j} \ell(\|x\|)}{P_{i} \ell(r_0)}} \\
&= \exp\Bigg(-\frac{\theta_i W}{P_{i} \ell(r_0)} - \sum_{j=1}^M \! \sum_{(x,P_{x}) \in \widetilde{\Phi}_j  \setminus (x_0,P_{x_0})} \!\!\!\!\!\!\!\!\!\!\!\!\!\!\!\!\!\ln\Big(1 + \frac{\theta_i P_{j} \ell(\|x\|)}{P_{i} \ell(r_0)}\Big)\Bigg) \\
&= \exp\Bigg(-\frac{\theta_i W}{P_{i} \ell(r_0)} - \sum_{j=1}^M \sum_{(x,P_{x}) \in \widetilde{\Phi}_j  \setminus (x_0,P_{x_0})} \!\!\!\!\!\!\!\!\!\!\!\!\!\!\!\!\!\tilde{\ell}_{ij}(r)\,\,\Bigg), \label{eqn:ps1}
\end{aligned}
\end{equation*}
where \(\tilde{\ell}_{ij}(r)\) is defined in (\ref{eqn:tilde_ell_ij}), which is non-increasing, monotonic, and bounded.

Since \(\widetilde{\Phi}_j\) is \((\sigma_j, \rho_j, \nu_j)\)-ball regulated with \(\sigma_j = P_j\), \(\rho_j = \frac{2\pi P_j}{\sqrt{12}H_j}\), and \(\nu_j = \frac{\pi P_j}{\sqrt{12}H_j^2}\), by the equivalence established in Theorem \ref{thm:equ}, \(\widetilde{\Phi}_j\) is also \((\sigma_j, \rho_j, \nu_j)\)-shot-noise regulated. 

Therefore, combining this with (\ref{eqn:R_inf_bound}) and (\ref{eqn:Ali}), we obtain the following bound:
\begin{equation}
\sum_{(x,P_{x}) \in \widetilde{\Phi}_j  \setminus (x_0,P_{x_0})} \!\!\!\!\!\!\!\!\!\!\!\!\!\!\!\tilde{\ell}_{ij}(r) \leq A_{\tilde{\ell}_{ij}}, \label{eqn:sum_bound}
\end{equation}
where \(A_{\tilde{\ell}_{ij}}\) is defined in (\ref{eqn:Ali_ij}).

Substituting (\ref{eqn:sum_bound}) into (\ref{eqn:ps1}), we derive the lower bound given in (\ref{eqn:Ps_bound1}). 
This derivation demonstrates that the SINR exceeds the threshold \(\theta_i\) with a probability constrained by noise and interference, thereby providing a lower bound on the link success probability for \(i\)-th class users by considering the combined effects of power control and shot-noise regulation.
\end{proof}

Theorem \ref{thm:linksuccbound} provides lower bounds on the link success probability for each of the \(M\) classes of links in the network. To derive an overall lower bound on the link success probability across all links, we define the minimum link success probability, denoted as \(P_\mathrm{s,min}\), as:
\begin{equation}
    P_\mathrm{s,min} \triangleq \min_{i \in [1, M]} P_\mathrm{s}(\theta_i \mid \widetilde{\Phi}).
\end{equation}

Using this definition, we present the following corollary, which establishes a lower bound on the link success probability across all links in the network.

\begin{corollary}
\label{cor:minlinksuccbound}
In a heterogeneous network scenario modeled by \(M\) MHCPPs, under Rayleigh fading conditions, the overall link success probability across all links satisfies the following inequality:
\begin{equation}
    P_\mathrm{s,min} \geq \min_{i \in [1, M]} \exp\Bigg(-\frac{\theta_i W}{P_i \ell(r_0)} - \sum_{j=1}^M A_{\tilde{\ell}_{ij}}\Bigg), \quad \mathbb{P}\text{-a.s.}, \label{eqn:Psmin_bound1}
\end{equation}
where \(A_{\tilde{\ell}_{ij}}\) is defined in (\ref{eqn:Ali_ij}).
\end{corollary}
\begin{proof}
    The result is directly derived from Theorem \ref{thm:linksuccbound}. By applying the minimization operation to both sides of inequality (\ref{eqn:Ps_bound1}) over all \(i \in [1, M]\), we obtain the desired bound.
\end{proof}

\begin{example}
Theorem \ref{thm:linksuccbound} is applicable to a general path loss function \(\ell(\cdot)\). As a specific example, when the path loss function is given by \(\ell(r) = \min\{1, r^{-\alpha}\}\), we can explicitly calculate \(A_{\tilde{\ell}_{ij}}\) as
\begin{equation}
\begin{aligned}
    &A_{\tilde{\ell}_{ij}}= \sigma_j \ln\left(1 + \frac{P_j \theta_i}{P_i\ell(r_0)}\right) \\
    &\quad + \rho_j \frac{\alpha P_j \theta_i}{(\alpha - 1) P_i \ell(r_0)} \, {}_2F_1\left(1 - \frac{1}{\alpha}, 1; 2 - \frac{1}{\alpha}; -\frac{P_j \theta_i}{P_i \ell(r_0)}\right) \\
    &\quad + \nu_j \frac{\alpha P_j \theta_i}{(\alpha - 2) P_i \ell(r_0)} \, {}_2F_1\left(1 - \frac{2}{\alpha}, 1; 2 - \frac{2}{\alpha}; -\frac{P_j \theta_i}{P_i \ell(r_0)}\right),
\end{aligned}
\label{eqn:A_ell_i_j}
\end{equation}
where \({}_2F_1(\cdot)\) denotes the Gauss hypergeometric function.

To arrive at this result, we substitute \(\ell(r) = \min\{1, r^{-\alpha}\}\) into the general expression of \(A_{\tilde{\ell}_{ij}}\) in (\ref{eqn:Ali_ij}), which yields:
\begin{equation}
\begin{aligned}
    &A_{\tilde{\ell}_{ij}} = \sigma_j \ln\left(1 + \frac{P_j \theta_i}{P_i\ell(r_0)}\right) 
    + \rho_j \ln \left(1 + \frac{P_j \theta_i}{P_i\ell(r_0)}\right) \\
    &\quad + \rho_j \int_{1}^{\infty} \ln \left(1 + \frac{P_j \theta_i r^{-\alpha}}{P_i\ell(r_0)}\right) \mathrm{d}r \\
    &\quad + \nu_j \ln \left(1 + \frac{P_j \theta_i}{P_i\ell(r_0)}\right) 
    + 2\nu_j \int_{1}^{\infty} r \ln \left(1 + \frac{P_j \theta_i r^{-\alpha}}{P_i\ell(r_0)}\right) \mathrm{d}r.
\end{aligned}
\end{equation}

The first integral term can be evaluated as
\begin{equation}
\begin{aligned}
    &\int_{1}^{\infty} \ln \left(1 + \frac{P_j \theta_i r^{-\alpha}}{P_i \ell(r_0)}\right) \mathrm{d}r 
    = -\ln \left(1 + \frac{P_j \theta_i}{P_i \ell(r_0)}\right) \\
    &\quad + \frac{\alpha P_j \theta_i}{(\alpha-1) P_i \ell(r_0)} {}_{2}F_{1}\left(1 - \frac{1}{\alpha}, 1; 2 - \frac{1}{\alpha}; -\frac{P_j \theta_i}{P_i \ell(r_0)}\right).
\end{aligned}
\end{equation}
    
Similarly, the second integral term evaluates to
\begin{equation}
\begin{aligned}
    &\int_{1}^{\infty} r \ln \left(1 + \frac{P_j \theta_i r^{-\alpha}}{P_i \ell(r_0)}\right) \mathrm{d}r 
    = -\frac{1}{2} \ln \left(1 + \frac{P_j \theta_i}{P_i \ell(r_0)}\right) \\
    &\quad + \frac{\alpha P_j \theta_i}{2(\alpha-2) P_i \ell(r_0)} {}_{2}F_{1}\left(1 - \frac{2}{\alpha}, 1; 2 - \frac{2}{\alpha}; -\frac{P_j \theta_i}{P_i \ell(r_0)}\right).
\end{aligned}
\end{equation}
    
Substituting these results back into the expression for \(A_{\tilde{\ell}_{ij}}\) completes the derivation of (\ref{eqn:A_ell_i_j}).
\end{example}

Theorem \ref{thm:linksuccbound} provides a deterministic lower bound for the link success probability in a large wireless network modeled by \(M\) MHCPPs, ensuring a minimum success probability for all links. 

However, when the SINR threshold for different tiers are the same, i.e. \(\theta_i=\theta\) for all \(i\), a simpler and looser lower bound can be derived with fewer steps, as shown in the following theorem.
\begin{corollary}
\label{cor:linksuccbound2}
In a heterogeneous network scenario with \(M\) MHCPPs and Rayleigh fading, a simplified lower bound for the link success probability \(P_{\mathrm{s}}(\theta_i \mid \widetilde{\Phi})\) of the \(i\)-th class is:
\begin{equation}
    P_{\mathrm{s}}(\theta_i \mid \widetilde{\Phi}) \geq \exp\left(-\frac{\theta_i (W + A_{\ell})}{P_{x_0} \ell(r_0)} + \theta_i\right), \quad \mathbb{P}\text{-a.s.},
\end{equation}
where \(A_{\ell}\) is defined in equation (\ref{eqn:Ali}), with \(\sigma = \sum_{i=1}^{M} P_i\), \(\rho = \sum_{i=1}^{M} \frac{2\pi P_i}{\sqrt{12}H_i}\), and \(\nu = \sum_{i=1}^{M} \frac{\pi P_i}{\sqrt{12}H_i^2}\).
\end{corollary}

\begin{proof}
Since \(\widetilde{\Phi}\) is the superposition of \(M\) MHCPPs with transmit power marks, where each individual marked point process \(\widetilde{\Phi}_i\) is \(\left(P_i, \frac{2\pi P_i}{\sqrt{12}H_i}, \frac{\pi P_i}{\sqrt{12}H_i^2}\right)\)-ball regulated, we apply Lemma \ref{lem:superreg} to conclude that the superposed marked point process \(\widetilde{\Phi}\) is strongly \(\left({\sigma}, {\rho}, {\nu}\right)\)-ball regulated, where
\begin{equation}
{\sigma} = \sum_{i=1}^{M} P_i, \quad {\rho} = \sum_{i=1}^{M} \frac{2\pi P_i}{\sqrt{12}H_i}, \quad \text{and} \quad {\nu} = \sum_{i=1}^{M} \frac{\pi P_i}{\sqrt{12}H_i^2}.
\end{equation}

Using Theorem \ref{thm:equ}, we then conclude that the superposed marked point process \(\widetilde{\Phi}\) is strongly \(\left(\sigma, \rho, \nu\right)\)-shot-noise regulated. 

Finally, applying Theorem \ref{thm:boundnofading_1tier} yields the simplified lower bound for the link success probability as stated.
\end{proof}

\subsection{Requirement-based Power Control}
To ensure that the requirements of different link classes are met, we consider a power control scheme that allocates distinct transmit powers \(P_{i}\) to the transmitters of the \(i\)-th link class based on their respective SINR requirements \(\theta_i\).

This power control strategy is designed to meet the quality of service requirements of each user class by distributing power uniformly within the same class, while varying it across classes according to their SINR thresholds. Classes with higher SINR requirements are allocated more power to meet their stringent quality of service needs. This approach defines our requirement-based power control scheme.

\begin{definition}\label{power-reg}
($\beta$-Power Control).
For a given constant \(\beta > 0\), the power \(P\) allocated to a link with a target SINR threshold \(\theta\) is defined by:
\begin{equation}
    P = \ln(1 + \beta \theta).
\end{equation}

This allocation formula applies to all transmitters for any given SINR threshold \(\theta\).
\end{definition}

According to Definition \ref{power-reg}, the power allocated to the \(i\)-th class is determined by its designated SINR threshold \(\theta_i\). Specifically, the allocated power is given by:
\begin{equation}
    P_i = \ln(1 + \beta \theta_i),  \label{eqn:power_alloc}
\end{equation}
where \(\theta_i\) represents the SINR requirement necessary for users in the \(i\)-th class to achieve satisfactory service quality. This power control scheme is motivated by Shannon's theorem, which indicates that an increase in the SINR threshold corresponds to a logarithmic increase in the achievable data rate.

Therefore, the transmit power for each transmitter is calculated according to the SINR thresholds assigned to different user classes. Users with higher SINR requirements are allocated more power, while those with lower requirements receive less. The primary challenge lies in determining whether this power control scheme can improve performance lower bounds or guarantee enhanced performance for all links in the wireless network.

In the following discussion, we explore the impact of the $\beta$-power control strategy on network performance. We will examine how this strategy influences the distribution and utilization of power across different classes of users within the wireless network. Additionally, we aim to derive a lower bound on the link success probability for networks with multi-requirement users, demonstrating how strategic power control can enhance overall network success probability.

The total power consumption for all transmitters within the wireless network is given by
\begin{equation}
\begin{aligned}
    P_\mathrm{total}(o,r)\leq\sum_{i=1}^M\left(1+\frac{2\pi r}{\sqrt{12}H_i} +\frac{\pi r^2}{\sqrt{12}H_i^2} \right)\mathrm{ln}(1+\beta\theta_i).
\end{aligned}
\end{equation}

By substituting the transmit power (\ref{eqn:power_alloc}) into Theorem \ref{thm:linksuccbound}, we derive the following corollary, which provides the lower bound for the link success probability under the \(\beta\)-power control strategy.

\begin{corollary}
\label{cor:linksuccbound_powercont}
In a heterogeneous network scenario modeled by \(M\) MHCPPs and considering Rayleigh fading, the link success probability \(P_{\mathrm{s}}(\theta_i \mid \widetilde{\Phi})\) for the \(i\)-th class with the \(\beta\)-power control strategy is given by
\begin{equation}
    P_\mathrm{s}(\theta_i \mid \widetilde{\Phi}) \geq \exp\Bigg(-\frac{\theta_i W}{\ln(1 + \beta \theta_i) \ell(r_0)} - \sum_{j=1}^M A_{\tilde{\ell}_{ij}}\Bigg), ~ \mathbb{P}\text{-a.s.}, \label{eqn:Ps_bound1_powercont}
\end{equation}
where \(A_{\tilde{\ell}_{ij}}\) is defined in (\ref{eqn:Ali_ij}), and \(\tilde{\ell}_{ij}(r)\) is given by
\begin{equation}
    \tilde{\ell}_{ij}(r) \triangleq \ln\left(1 + \frac{\theta_i \ln(1 + \beta \theta_j) \ell(r)}{\ln(1 + \beta \theta_i) \ell(r_0)}\right). \label{eqn:tilde_ell_ij_powercont}
\end{equation}
\end{corollary}

\section{Numerical Results}
In this section, we delve into the performance evaluation of a wireless network serving two distinct classes of users, categorized according to their SINR thresholds: high-requirement and low-requirement. We perform simulations to assess the practical effectiveness of our theoretical models and to compare the simulated outcomes with the theoretically derived lower bounds of link success probability.
The simulation environment is configured with transmitters arranged in a triangular lattice within a circular region of radius $r = 100$, adhering to hardcore distances $H_1$ for high-requirement users and $H_2$ for low-requirement users. % {We choose the triangular lattice instead of MHCPP for simulation because MHCPP randomly removes nodes based on the hardcore distance, making it more difficult to simulate accurately the worst-case reliability. In contrast, the triangular lattice allows nodes to be densely packed within the network, better representing the worst-case interference scenarios.} 
{We choose the triangular lattice over MHCPP for simulations because it is more difficult to accurately simulate the worst-case reliability with MHCPP. In contrast, the triangular lattice allows nodes to be densely packed within the network, better representing the worst-case interference scenarios.} This configuration ensures that the distribution of transmitters corresponding to users of different requirements is mutually independent and unaffected by users of other requirements. However, the 
$i$-th class of links are still subject to interference from the transmitters corresponding to users of other classes. The stationary marked point process governing the locations of the transmitters $\Phi_i$ is strongly regulated by a $(1, 2\pi / (\sqrt{12}H_i), \pi / (\sqrt{12}H_i^2))$-ball regulation model, which provides a basis for calculating theoretical lower bounds and conducting rigorous simulations.
The simulation parameters include SINR thresholds $\theta_1$ and $\theta_2$ for high-requirement and low-requirement users, respectively, assuming $\theta_1>\theta_2>0$. Additional parameters are set as follows: power factor $\beta=1$, noise power $W=0$ and a fixed distance $r_0=1$ between users and their dedicated transmitters. The path loss function is defined as $\ell(r)={\rm min}\lbrace{1, r^{-\alpha}}\rbrace$. In subsequent analyses, we set the pass loss exponent, denoted as $\alpha$, to $4$. We define the reference threshold as $\theta$, setting $\theta_1$ equal to $\theta$ and $\theta_2$ to $\theta-3\rm dB$.
% \iffalse
% \begin{figure}[!htb]
%     \centering
%     \subfigure[$\alpha=3$]{
%         \begin{minipage}{0.5\textwidth}
%             \centering
%             \includegraphics[width=1\textwidth]{changet3.eps}
%         \end{minipage}
%     }
%     \subfigure[$\alpha=4$]{
%         \begin{minipage}{0.5\textwidth}
%             \centering
%             \includegraphics[width=1\textwidth]{changet4.eps}
%         \end{minipage}
%     }
%     \caption{Coverage Probability vs. the Changes of of Users' Thresholds}
%     \label{fig-pcov-t}
% \end{figure}
% \fi

\begin{figure}[htbp]
\centering
\includegraphics[width=0.5\textwidth]{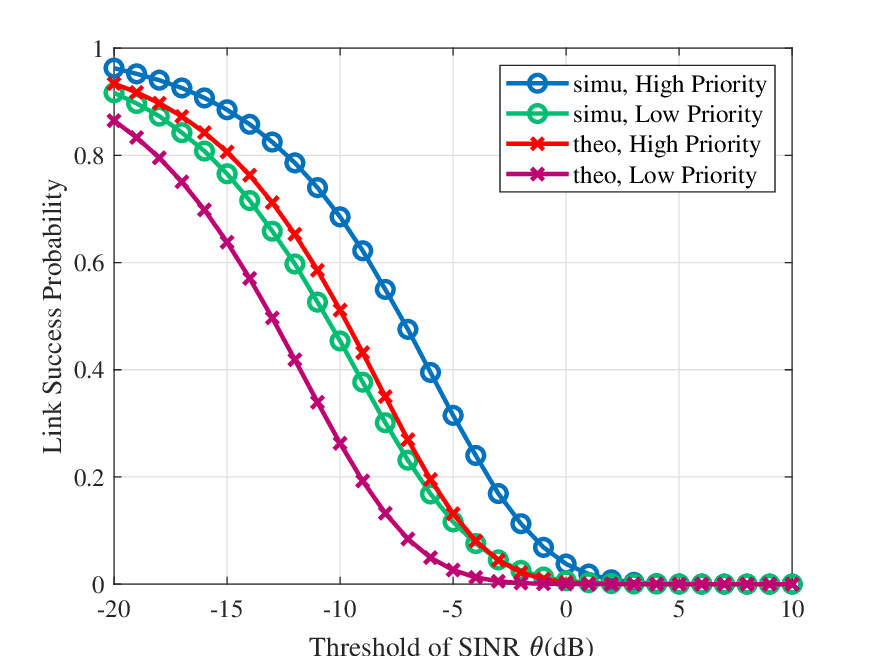}
\caption{Link success probability vs. the changes of the users' thresholds. }
\label{fig-changetheta}
\end{figure}

Fig. \ref{fig-changetheta} illustrates the variation in the lower bounds of link success probability for two classes of users as their SINR thresholds ($\theta_1$ and $\theta_2$) change. This figure includes both simulation results and theoretical calculations for high-requirement users and low-requirement users, respectively. For users with the same requirement, the vertical gap represents the differences in the link success probability lower bound between simulation results and theoretical calculations. This discrepancy was mainly influenced by the point process modeling the locations of transmitters and the path loss function. {Theorem \ref{thm:linksuccbound} provides a theoretical lower bound on the link success probability, while the achievability is beyond the scope of the paper.} The horizontal gap indicates the difference in SINR threshold required to achieve the desired link success probability. In particular, when the link success probability exceeds $0.8$, the horizontal deviations between the simulation results and the theoretical lower bounds are less than $3\,\mathrm{dB}$ for high-requirement users and $2\,\mathrm{dB}$ for low-requirement users. As the SINR threshold increases, the lower bound of link success probability decreases. This indicates that higher user demands make it increasingly difficult to meet all user requirements. Low-requirement users suffer from lower link success probability than high-requirement users because their dedicated transmitters are allocated less power, leading to a lower SINR, thus necessitating a lower threshold to satisfy their requirements.

\begin{figure}[htbp]
\centering
\includegraphics[width=0.5\textwidth]{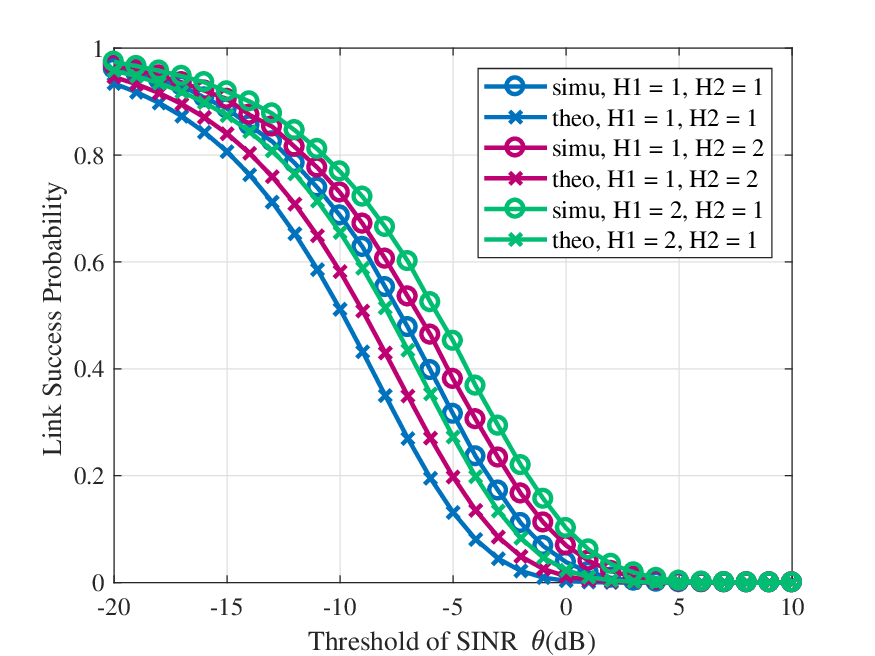}
\caption{Link success probability with different hardcore distances}
\label{fig-changedistance}
\end{figure}

Fig. \ref{fig-changedistance} shows the impact of different hardcore distances on the link success probability lower bound for high-requirement users. This analysis considers scenarios where transmitters for high-requirement and low-requirement users are positioned at distinct hardcore distances. Specifically, when the hardcore distances $H_1$ and $H_2$ are set to $1$, the link success probability lower bound reaches its minimum. This decrease is attributed to the increase in the density of transmitters, which intensifies interference within the wireless network. Furthermore, when the hardcore distances are set as $H_1 = 2$ and $H_2 = 1$, a higher link success probability lower bound is observed. This improvement can be attributed not only to the reduced density of transmitters regulated by the hardcore distance, but also to the power control strategy, which allocates less power to low-requirement transmitters, thereby reducing the overall level of interference within the wireless network.

\begin{figure}[htbp]
\centering
\includegraphics[width=0.5\textwidth]{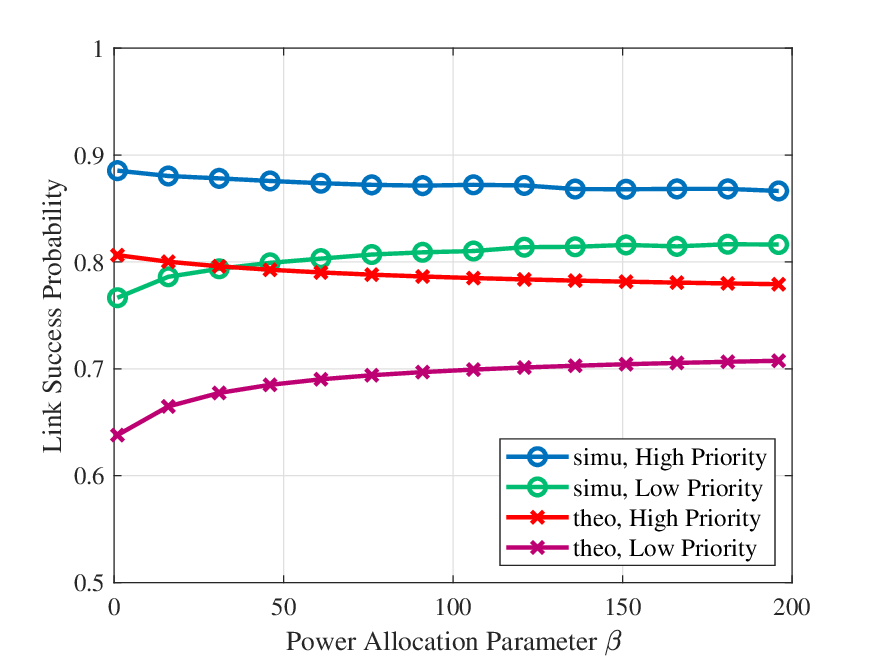}
\caption{Link success probability vs. the changes of \(\beta\)}
\label{fig-changebeta}
\end{figure}

{In Fig. \ref{fig-changebeta}, the impact of the parameter $\beta$ within the $\beta$-power control strategy on link success probability lower bound is illustrated, with $\beta$ varying from $1$ to $200$, $\theta$ set at $-15\,\mathrm{dB}$, and the hardcore distances $H_1$ and $H_2$ are set to $1$.} It shows a clear increase in the link success probability lower bound for low-requirement users as $\beta$ increases, which coincides with a decrease in link success probability lower bound for high-requirement users. Furthermore, there is no optimal $\beta$ value that achieves higher link success probability lower bound for both high-requirement and low-requirement users simultaneously. 

\begin{figure}[htbp]
\centering
\includegraphics[width=0.5\textwidth]{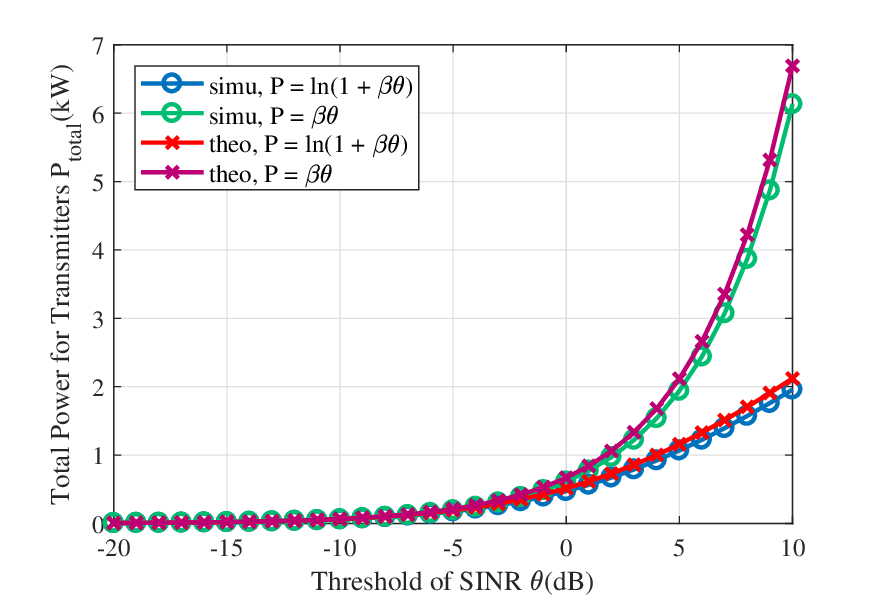}
\caption{Total Power for Transmitters vs. the Changes of \(\theta\)}
\label{fig-compTotalPower}
\end{figure}

\begin{figure}[htbp]
\centering
\includegraphics[width=0.5\textwidth]{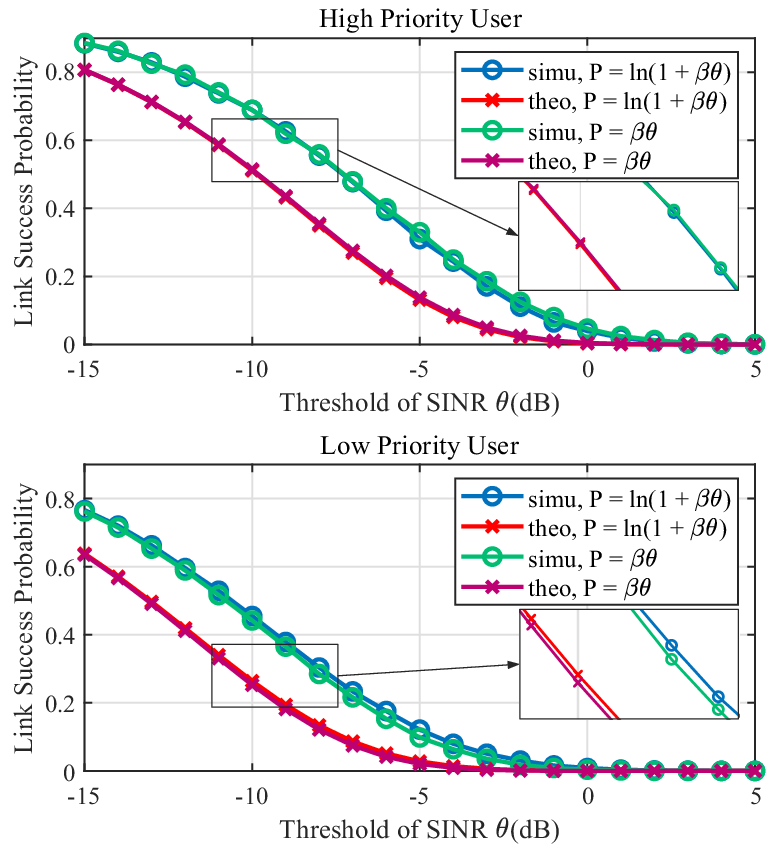}
\caption{Link success probability lower bound vs. the changes of \(\theta\)}
\label{fig-compReliability}
\end{figure}

Fig. \ref{fig-compTotalPower} examines the variations in total power consumption of the transmitters with changes in the SINR threshold $\theta$, comparing two power control strategies: $P=\ln{\left(1+\beta\theta\right)}$ as mentioned in Def. \ref{power-reg} and $P=\beta\theta$. The theoretical results provide the upper bound on $P_{\rm total}(o,r)$, as derived from (\ref{totalPower}). The proposed power control strategy consumes less total power under all SINR thresholds, and the power difference continues to increase as the SINR threshold increases. This implies that utilizing the $\beta$-power control strategy is more effective to reduce power consumption. 
Fig. \ref{fig-compReliability} presents a comparison of the link success probability lower bounds for high-requirement and low-requirement users under two different power control strategies, given equal total power for transmitters { and with \(\beta\) set to \(1\). To better highlight the differences, magnified views of the curves around the SINR threshold \(\theta = -10\,\mathrm{dB}\) are provided.} The figures indicate that using the allocation strategy $P=\ln{\left(1+\beta\theta\right)}$ slightly 
enhances the link success probability lower bound for low-requirement users, with no marked impact on high-requirement users.

\section{Conclusion}
This work introduces a novel regulatory framework to constrain total transmit power and aggregated interference in large-scale wireless networks. As an example, we utilize the Matérn II method to demonstrate how these regulations can be algorithmically implemented. Our work provides computable bounds on link performance suitable for various user requirements, establishing a robust analytical framework for evaluating network efficiency. Monte Carlo simulations validated our approach across different user classes, underscoring the practical relevance and efficacy of our regulatory strategies in real-world settings.

Looking ahead, our research will explore the integration of queuing theory in the time domain with spatial-temporal network calculus to enhance service quality across different requirement levels. {While the current framework uses spatial regulation to ensure deterministic performance, it does not consider interference dynamics. Future work will integrate spatial regulation with dynamic networks to better reflect real-world environments.} We also plan to assess the impacts of environmental factors such as shadowing, exposure to radio frequency electromagnetic fields, and strategies such as dynamic power control and load balancing on network performance, aiming to refine and extend the performance guarantees essential for future wireless network deployments. {Investigating the achievability of performance bounds for all links in arbitrarily large wireless networks remains an open and meaningful direction for future research.}

%%%%%%
%% Appendix:
%% If needed a single appendix is created by
%%
%\appendix
%%
%% If several appendices are needed, then the command
%%
% \appendices
%%
%% in combination with further \section commands can be used.
%%%%%%

%%%%%%
%% To balance the columns at the last page of the paper use this
%% command:
%%
%\enlargethispage{-1.2cm} 
%%
%% If the balancing should occur in the middle of the references, use
%% the following trigger:
%%
% \IEEEtriggeratref{4}
%%
%% which triggers a \newpage (i.e., new column) just before the given
%% reference number. Note that you need to adapt this if you modify
%% the paper.  The "triggered" command can be changed if desired:
%%
%\IEEEtriggercmd{\enlargethispage{-20cm}}
%%
%%%%%%

%%%%%%
%% References:
%% We recommend the usage of BibTeX:
%%
\bibliographystyle{IEEEtran}
\bibliography{bibfile.bib}
%%
%% where we here have assumed the existence of the files
%% definitions.bib and bibliofile.bib.
%% BibTeX documentation can be obtained at:
%% http://www.ctan.org/tex-archive/biblio/bibtex/contrib/doc/
%%%%%%

\end{document}